\documentclass{ieeeaccess}

\usepackage{cite}
\usepackage{amsmath,amssymb,amsfonts}
\usepackage{algorithmic}
\usepackage{graphicx,color}
\usepackage{textcomp}
\def\BibTeX{{\rm B\kern-.05em{\sc i\kern-.025em b}\kern-.08em
    T\kern-.1667em\lower.7ex\hbox{E}\kern-.125emX}}
\AtBeginDocument{\definecolor{ojcolor}{cmyk}{0.93,0.59,0.15,0.02}}
\def\OJlogo{\vspace{-4pt}\hskip-4pt\includegraphics[height=18pt]{OJcoms.png}}

\usepackage{amsmath,amssymb,amsfonts,mathrsfs, bm}
\usepackage{algorithm}
\usepackage{array}
\usepackage[caption=false,font=normalsize,labelfont={scriptsize,sf},textfont=sf]{subfig}
\usepackage{textcomp}
\usepackage{stfloats}
\usepackage{url}
\usepackage{verbatim}
\usepackage{graphicx}
\usepackage{cite}
\usepackage{booktabs}
\usepackage{multirow}

\usepackage{mathrsfs}
\usepackage{balance}

\usepackage{setspace}
\usepackage{longtable, threeparttable, tabularx}
\usepackage{makecell}

\usepackage{array}
\usepackage{booktabs}
\usepackage{tabularx}

\newcolumntype{Y}{>{\raggedright\arraybackslash}X}

\allowdisplaybreaks[4]

\newtheorem{corollary}{Corollary}

\newtheorem{remark}{Remark}

\usepackage{xurl}
\let\code\path

\usepackage{xurl}
\usepackage[hidelinks,breaklinks=true]{hyperref}

\begin{document}
\history{Date of current version: July 24th, 2026.}
\doi{XXX}


\title{A Survey of Typical-Cell Volume Distributions in Poisson--Voronoi and Poisson--Delaunay Tessellations: Analytical Theory, High-Dimensional Limits, and Wireless Applications}

\author{MINGHUA XIA\authorrefmark{1}, \IEEEmembership{Senior Member, IEEE}, TIAN SHI\authorrefmark{1}, WENKUN WEN\authorrefmark{2} \IEEEmembership{Member, IEEE}}
\address[1]{School of Electronics and Information Technology, Sun Yat-sen University, Guangzhou 510006, China (e-mail: xiamingh@mail.sysu.edu.cn, shit26@mail2.sysu.edu.cn).}
\address[2]{The R\&D Department, Techphant Technologies Company Ltd., Guangzhou 510310, China (email: wenwenkun@techphant.net).}

\markboth
{MINGHUA XIA \headeretal: A Survey of Typical-Cell Volume Distributions in Poisson--Voronoi and Poisson--Delaunay Tessellations}
{MINGHUA XIA \headeretal: A Survey of Typical-Cell Volume Distributions in Poisson--Voronoi and Poisson--Delaunay Tessellations}

\corresp{Corresponding author: Minghua~Xia (e-mail: xiamingh@mail.sysu.edu.cn)}

\begin{abstract}
Random spatial tessellations generated by point processes provide fundamental models for proximity relations, space partitioning, and local geometry in stochastic systems. Among them, Poisson--Voronoi and Poisson--Delaunay tessellations induced by homogeneous Poisson point processes form a canonical dual pair widely used in stochastic geometry, computational geometry, spatial statistics, and wireless-network analysis. In many applications, the typical-cell volume distribution provides an important geometric input for modeling coverage regions, traffic load, clustering structure, connectivity, and other system-level characteristics. Despite extensive study, the literature remains fragmented and analytically asymmetric. For Poisson--Voronoi cell volumes, exact integral representations are available in certain planar settings, while a recent scale--shape factorization provides an exact general-dimensional representation with conditional Gamma structure. However, the normalized shape laws and unbounded facet-count mixture remain implicit, and tractable unconditional closed-form distributions are still unavailable. Practical modeling therefore continues to rely largely on simulations, moment-based characterizations, and empirical approximations. By contrast, Poisson--Delaunay simplex volumes admit directly evaluable, dimension-explicit PDFs, CDFs, and moment formulas derived through Mellin-transform analysis and Meijer's \(G\)-function representations. Motivated by this contrast, this paper presents a structured survey of typical-cell volume distributions in Poisson--Voronoi and Poisson--Delaunay tessellations. We review the principal analytical methodologies, synthesize established exact and approximate results, summarize emerging high-dimensional limit behavior, and discuss representative wireless-network applications, including load modeling, cooperative transmission, and three-dimensional network architectures. We also identify open problems involving unconditional Poisson--Voronoi distributions, non-Poisson spatial models, data-driven geometric inference, and dimension-aware network modeling.
\end{abstract}

\begin{keywords}
Stochastic geometry, Poisson point process, Voronoi tessellation, Delaunay tessellation, typical-cell volume distribution, survey, wireless networks.
\end{keywords}

\maketitle

\section{Introduction}
\label{sec:introduction}

\IEEEPARstart{R}{andom} spatial tessellations provide a mathematically rigorous framework for describing how space is partitioned by interacting nodes, seeds, or geometric generators. As such, they play a central role in stochastic geometry and have found broad applications in computational geometry, spatial statistics, machine learning, biology, materials science, transportation systems, and wireless communications~\cite{Okabe2000,Schneider2008,Chiu2013}. Among the many tessellation models, Voronoi tessellations and their dual Delaunay tessellations are particularly fundamental because they encode two complementary geometric principles: nearest-neighbor partitioning and adjacency connectivity~\cite{pournin2012voronoi}. When the underlying generators follow a homogeneous Poisson point process (PPP), the resulting Poisson--Voronoi tessellation (PVT) and Poisson--Delaunay tessellation (PDT) become canonical random models with strong invariance properties and substantial analytical tractability.

\IEEEpubidadjcol

\begin{figure*}[!t]
    \centering
    \includegraphics[width=0.99\textwidth]{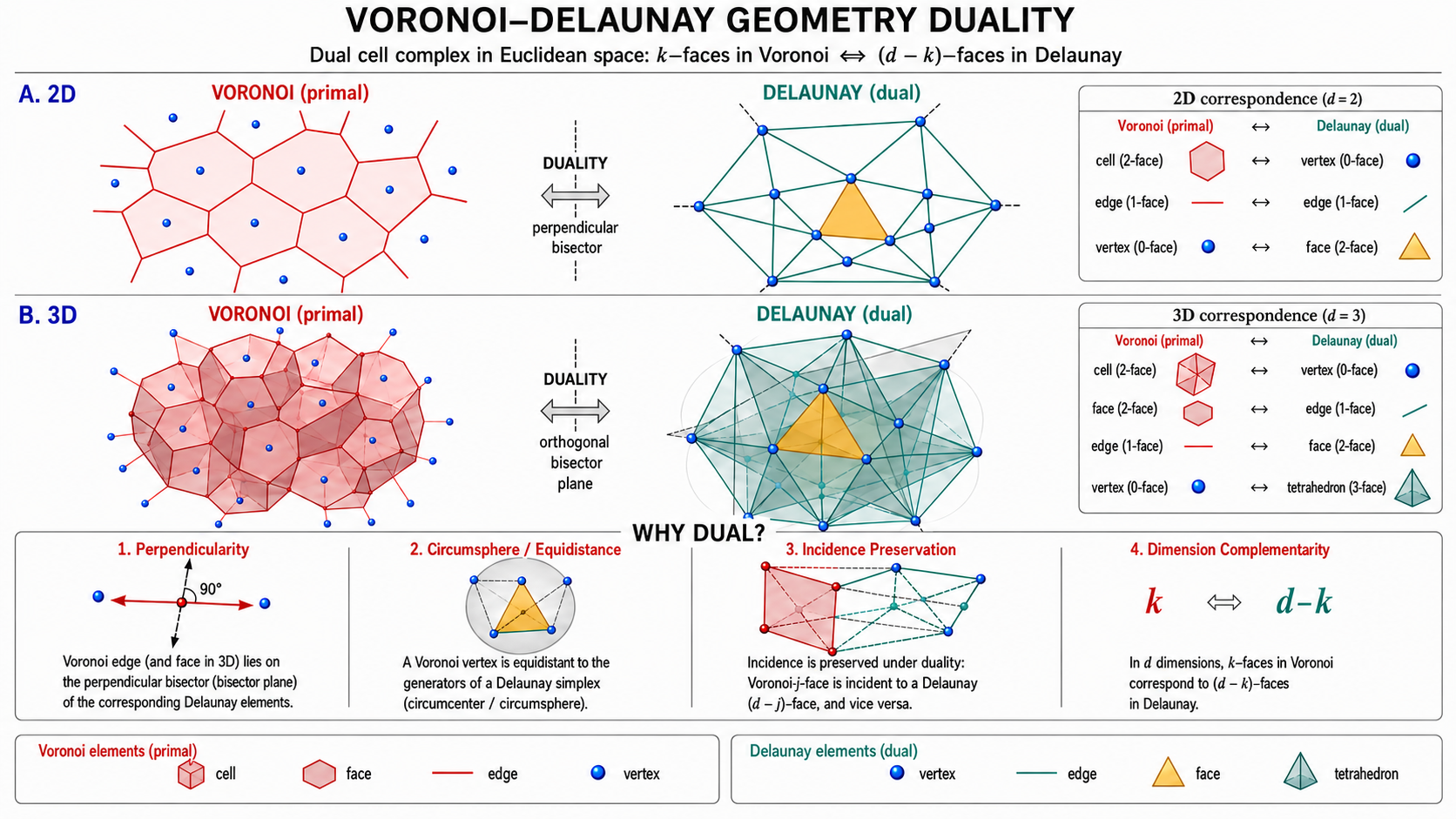}
    \caption{Geometric duality between Voronoi and Delaunay tessellations in two and three dimensions. In two dimensions, Voronoi cells, edges, and vertices correspond respectively to Delaunay vertices, edges, and triangular faces. In three dimensions, Voronoi cells, faces, edges, and vertices correspond respectively to Delaunay vertices, edges, triangular faces, and tetrahedra. All vertices are represented by blue dots, while Voronoi elements are shown in red and Delaunay elements are shown in teal/yellow.}
    \label{fig:voronoi_delaunay_duality}
\end{figure*}

Fig.~\ref{fig:voronoi_delaunay_duality} illustrates the geometric duality between Voronoi and Delaunay tessellations in two and three dimensions. In the planar case, Voronoi cells, edges, and vertices are dual to Delaunay vertices, edges, and triangular faces, respectively. In the spatial case, this correspondence follows the dimension-complementarity rule, whereby a \(k\)-face of the Voronoi tessellation corresponds to a \((d-k)\)-face of the Delaunay tessellation in dimension \(d\). Consequently, in three dimensions, Voronoi cells, faces, edges, and vertices correspond to Delaunay vertices, edges, faces, and tetrahedral cells, respectively. This duality is induced by the perpendicular-bisector and circumsphere properties and preserves the incidence relations between primal and dual elements. It therefore provides a geometric foundation for relating local adjacency structures to the volumetric properties of typical cells in planar and spatial tessellations. For numerical construction, the required Voronoi--Delaunay structures can be generated using standard computational-geometry algorithms, such as the Bowyer--Watson incremental insertion algorithm, Fortune's sweep-line method, and divide-and-conquer strategies~\cite{Hjelle2006}.

These two tessellation families have become especially influential in modern wireless-network analysis \cite{Baccelli2009,Baccelli2010,Haenggi2012}. In cellular systems, Voronoi cells naturally describe coverage regions under nearest-base-station association and are closely related to user load, traffic imbalance, handover behavior, and resource allocation~\cite{6042301}. In contrast, Delaunay structures characterize neighborhood relations among infrastructure nodes and provide useful geometric abstractions for cooperative transmission, clustering, relaying, mesh connectivity, and backhaul design~\cite{8358978}. Beyond conventional terrestrial networks, Voronoi and Delaunay models have also been used in uncrewed aerial vehicle (UAV) systems~\cite{10720695,11346534}, low-altitude networks~\cite{11311373}, satellite constellations~\cite{10531691,11159552}, and emerging three-dimensional communication architectures~\cite{8533634,9151343}.

Within these applications, one of the most fundamental yet comparatively underexplored quantities is the \emph{typical-cell volume distribution}. The volume of a typical Voronoi cell measures the random service region associated with a node, while the volume of a typical Delaunay simplex reflects local spatial density and adjacency geometry. Beyond their mean values, the variability and tail behavior of these volumes provide information about atypically small or large geometric regions that cannot be captured by average-volume analysis alone. These distributions are therefore relevant to many engineering metrics, including coverage-area variability, cell-load fluctuations, cooperation-cluster size, connectivity robustness, and spatial fairness. Although volume statistics are not standalone performance measures, they provide geometric inputs that can be combined with propagation, association, mobility, and resource-allocation models. Consequently, understanding typical-cell volume statistics is important not only from a geometric-probability perspective but also for modeling and optimizing large-scale wireless systems.

\IEEEpubidadjcol

Despite decades of research, the current literature on typical-cell volume distributions remains fragmented and exhibits a pronounced analytical asymmetry. For Poisson--Voronoi tessellations, exact closed-form volume distributions are generally unavailable except in one dimension. In the planar case, Calka derived exact integral representations for the distributions of the principal geometric characteristics of the typical cell, including its area and perimeter~\cite{Calka2003}. These representations are mathematically exact, but they involve conditioning on the number of cell sides, multi-dimensional integrations over neighboring-point distances and angular spacings, and an infinite mixture over the side count for the unconditional area distribution. Very recent work~\cite{xia2026scaleshape} has further established an exact scale--shape factorization of the typical Poisson--Voronoi cell volume, including a conditional Gamma law for the Voronoi flower volume and exact mixture, transform, and moment identities. Nevertheless, a tractable unconditional cell-volume distribution remains unavailable. As a result, practical models still rely largely on Monte Carlo simulations, moment-based characterizations, Gamma-type approximations, and other empirically fitted distributions. By contrast, Poisson--Delaunay simplex volumes admit a substantially richer analytical treatment. Existing works have developed Mellin--Barnes integral representations and, more recently, exact probability density functions (PDFs), cumulative distribution functions (CDFs), and moment expressions in arbitrary dimensions through Mellin-transform and Meijer's $G$-function methods~\cite{Rathie1992,Muche1996,8358978}. This sharp contrast makes the dual pair of PVTs and PDTs particularly interesting from both mathematical and engineering viewpoints.

Meanwhile, the growing use of random tessellations in high-dimensional data analysis, geometric machine learning, random simplicial complexes, and large-scale computational modeling has stimulated interest beyond the conventional two- and three-dimensional settings. High-dimensional Voronoi and Delaunay structures arise naturally in search-space exploration, Bayesian optimization, topological data analysis, geometric interpolation, and machine learning, where cells and simplices characterize local neighborhoods and geometric relationships among data points~\cite{Mishra2023,Zhao2024}. As the ambient dimension increases, however, their combinatorial complexity grows rapidly, making exact construction and numerical integration increasingly challenging~\cite{Sikorski2024}. At the same time, their geometric characteristics and cell volumes may exhibit pronounced scaling and concentration phenomena that differ substantially between the two dual tessellations~\cite{AHL_2021,irlbeck2025pvtd}. Understanding these effects is important for assessing the validity of low-dimensional intuition, developing dimension-aware approximations, and designing computationally tractable models. These trends therefore motivate a broader examination of typical-cell volume laws, with particular emphasis on their dimension-dependent scaling, concentration behavior, asymptotic limits, and numerical characterization.

Motivated by the above observations, this article presents a structured narrative survey of typical-cell volume distributions in Poisson--Voronoi and Poisson--Delaunay tessellations. Rather than treating the two tessellation families separately, we adopt a unified viewpoint centered on their geometric duality, contrasting levels of analytical tractability, and shared relevance to stochastic-geometry-based wireless modeling. The contribution of this survey lies in synthesizing, comparing, and interpreting classical and recent results, including analytical methodologies, high-dimensional behavior, and representative engineering applications. Analytical formulas and limit results are attributed to their original sources, while the derivations presented here connect or reformulate selected results within a common notation and comparative framework. The numerical results serve as illustrations and model comparisons, and the engineering discussion provides application-oriented interpretations of the surveyed geometric results rather than standalone wireless-system performance claims.

The main contributions of this survey are summarized as follows:

\begin{itemize}

\item \emph{Unified review of dual tessellation models:}
We provide a systematic comparison of Poisson--Voronoi and Poisson--Delaunay tessellations from the perspective of typical-cell volume statistics, highlighting their geometric duality and contrasting analytical properties.

\item \emph{Comprehensive synthesis of analytical methodologies:}
We review the principal techniques for characterizing volume distributions, including simulation benchmarking, moment-based analysis, parametric approximations, Mellin-transform methods, and Meijer's-$G$ representations.

\item \emph{Survey of dimension-dependent and high-dimensional results:}
We summarize known exact formulas, asymptotic behaviors, concentration effects, and emerging scaling laws for typical-cell volumes as the spatial dimension increases.

\item \emph{Wireless-network perspectives and design relevance:}
We discuss representative applications in cellular coverage, traffic load modeling, cooperative transmission, connectivity analysis, and three-dimensional wireless networks.

\item \emph{Future research roadmap:}
We identify open directions, including non-Poisson spatial models, dynamic tessellations, machine-learning-assisted geometric inference, and dimension-aware wireless-network design.

\end{itemize}

The remainder of the survey is organized as follows. Section~\ref{sec:method_landscape} reviews the main analytical methodologies for characterizing typical-cell volume distributions. Section~\ref{sec:PDFs_of_volumes} surveys known distributional results for Poisson--Voronoi and Poisson--Delaunay tessellations and derives high-dimensional limits. Section~\ref{sec:results} presents numerical illustrations and approximation comparisons. Section~\ref{sec:application} reviews representative wireless-network applications. Section~\ref{sec:future_research} discusses future opportunities and key open problems. Section~\ref{sec:conclusion} concludes the paper.

\emph{Notation:}
Throughout the paper, the standard Bachmann--Landau asymptotic notation is used.
For two positive functions $f(x)$ and $g(x)$, $f(x)=\mathcal{O}(g(x))$ as $x\to\infty$ if there exist constants $C>0$ and $x_0$ such that $f(x)\le C g(x)$ for all $x\ge x_0$.
We write $f(x)=\Theta(g(x))$ if $c_1 g(x)\le f(x)\le c_2 g(x)$ for some constants $c_1,c_2>0$ and all sufficiently large $x$, and $f(x)=o(g(x))$ if $\lim_{x\to\infty} f(x)/g(x)=0$.
Equality in distribution is denoted by $\stackrel{\mathrm{d}}{=}$, and the expectation and variance operators are denoted by $\mathbb{E}[\cdot]$ and $\operatorname{Var}(\cdot)$, respectively.
Convergence in probability and convergence in distribution are denoted by $\xrightarrow{P}$ and $\xrightarrow{\mathcal{D}}$, respectively.
The Gamma function is denoted by $\Gamma(\cdot)$.

%
\section{Methodological Landscape: From Simulation to Closed Forms}
\label{sec:method_landscape}

The literature on typical-cell volume distributions in PVTs and PDTs spans a wide spectrum of analytical depth. At one end, Monte Carlo simulation provides high-accuracy empirical benchmarks for typical-cell PDFs and moments. At the other end, exact closed-form PDFs and CDFs have been obtained for Poisson--Delaunay simplex volumes by combining raw-moment identities with inverse Mellin transforms and Meijer's-$G$ representations. Between these extremes lie several important methodological families, including parametric distribution fitting, moment-based characterization, exact integral representations, scale--shape factorization, finite-window correction, and asymptotic or limit-theorem-based analysis. This section summarizes these approaches and clarifies their respective advantages and limitations.

\subsection{Voronoi Cells: From Empirical Models to Exact Structural Factorization}

For PVTs, deriving a tractable closed-form PDF or CDF of the typical-cell volume remains a long-standing open problem, unresolved even in the classical low-dimensional cases $d=2$ and $d=3$. Since the early work of Kiang~\cite{kiang1966random} and subsequent simulation-based investigations, including those of Hinde and Miles~\cite{hinde1980monte} and Tanemura~\cite{tanemura2003statistical}, this problem has motivated a broad range of empirical and analytical approaches. Historically, most progress relied on Monte Carlo benchmarking, parametric fitting, and low-order moment characterization. Exact planar integral formulae later provided rigorous conditional representations, while recent work has established an exact scale--shape factorization in arbitrary dimensions~\cite{xia2026scaleshape}. Although this factorization does not yield a directly evaluable unconditional closed-form PDF, it identifies the probabilistic origin of the Gamma structure and reduces the remaining distributional difficulty to normalized shape laws and mixing over the effective-facet count. The principal methodological developments can therefore be grouped as follows.

\subsubsection{Simulation-Based Estimation}

Large-scale Monte Carlo simulations remain the primary means of generating empirical PDFs and benchmarking geometric statistics of typical Poisson--Voronoi cells. The influential Monte Carlo study of Hinde and Miles~\cite{hinde1980monte} helped establish simulation as a central methodology for examining typical-cell distributions. A widely cited later investigation by Tanemura~\cite{tanemura2003statistical} provided detailed numerical statistics for both $d=2$ and $d=3$. Simulation-based estimation can achieve high numerical accuracy and remains indispensable for validating analytical approximations, examining distribution tails, and studying dimensions for which few exact results are available. Nevertheless, it does not yield analytical distributions, and its computational cost increases rapidly with the spatial dimension and the required tail resolution.

\subsubsection{Parametric Distribution Fitting}

A common engineering approach is to approximate the empirical volume distribution using tractable parametric families, most notably Gamma and generalized-Gamma distributions. Influential examples include the Gamma-type models proposed by Weaire \emph{et al.}~\cite{weaire1986distribution} and the one-parameter approximation of Ferenc and N\'eda~\cite{ferenc2007size}. Such models provide compact PDFs and CDFs that can be incorporated into system-level analysis, with parameters adopted from empirical studies or estimated through moment matching and numerical fitting. Recent scale--shape factorization results provide a structural rationale for the recurring effectiveness of Gamma-type models by identifying an exact conditional Gamma component in the underlying Poisson--Voronoi geometry~\cite{xia2026scaleshape}. However, they also show that the cell-volume distribution is generally not itself Gamma. The accuracy of an approximation therefore continues to depend on its parameterization, dimension, fitting criterion, and tail fidelity.

\subsubsection{Moment-Based Approaches}

A complementary methodology characterizes the typical-cell volume through raw moments and moment-matching constraints. Standard references such as Okabe \emph{et al.}~\cite{Okabe2000} and Chiu \emph{et al.}~\cite{Chiu2013} summarize benchmark moment information for Poisson--Voronoi cells. More specific results include planar area moments obtained by Hayen and Quine~\cite{hayen2002areas} and the second volume moment of the three-dimensional typical cell derived by Muche and Ballani~\cite{MucheBallani2011}. These results provide rigorous benchmarks and support the calibration of Gamma-type approximations. Nevertheless, a finite collection of moments does not uniquely reveal the complete distributional structure, and deriving higher-order moments remains analytically difficult.

\subsubsection{Exact Planar Integral Formulae}

Calka~\cite{Calka2003} derived exact integral formulae for the principal geometric characteristics of the typical two-dimensional Poisson--Voronoi cell, including its area and perimeter. The distributions are represented conditionally on the number of cell sides through integrations over neighboring-point distances and angular configurations, while the unconditional law requires an infinite mixture over possible side counts. These formulae provide exact planar representations and reveal the role of the cell's combinatorial structure. However, their direct numerical evaluation is demanding, and they do not provide dimension-explicit closed-form PDFs or a general-dimensional separation of scale and shape.

\subsubsection{Exact Scale--Shape Factorization in Arbitrary Dimensions}

Recent work has established an exact scale--shape factorization of the Palm-typical Poisson--Voronoi cell volume in arbitrary dimensions~\cite{xia2026scaleshape}. The construction conditions on the effective-facet count and uses the Voronoi flower volume as the radial scale. This flower is the exclusion region whose emptiness ensures that a candidate configuration of neighboring nuclei defines the Palm cell. Conditional on the facet count, its volume is exactly Gamma distributed and independent of the normalized cell shape. The cell volume is therefore represented as the product of an independent Gamma scale and a bounded normalized shape factor.

Mixing over the shape factor and facet count yields exact distributional, transform, and moment representations. The factorization also identifies two sources of departure from a single Gamma law: shape variability within a fixed facet-count stratum and mixing across facet counts. It thereby provides a theoretical basis for Gamma-type approximations and a framework for studying the distribution tails. Nevertheless, the normalized shape laws and the unbounded facet-count mixture remain implicit. Thus, the result provides an exact structural characterization, but not a tractable unconditional closed-form PDF or CDF.

\subsubsection{Structural and Asymptotic Results}

Beyond distributional representations, rigorous results have been developed for structural and asymptotic properties of Poisson--Voronoi cells, including geometric functionals, combinatorial statistics, extreme-cell behavior, and limit theorems. Representative contributions include those of Calka~\cite{Calka_2002}, M{\o}ller~\cite{Moller1994}, and Schulte~\cite{SCHULTE2012285}. Recent work has further examined the high-dimensional geometry of the typical Poisson--Voronoi cell and established dimension-dependent limits for its principal geometric characteristics~\cite{irlbeck2025pvtd}. These results provide substantial theoretical insight into scaling and concentration phenomena but generally do not yield tractable finite-dimensional PDFs for the typical-cell volume.

\subsubsection{Finite-Window Analysis and Boundary Correction}

In practical applications and numerical experiments, observations are restricted to bounded domains, where cells intersecting the observation boundary are truncated and can bias empirical volume statistics. Finite-window corrections and boundary-aware models are therefore needed to distinguish typical-cell behavior from artifacts of the observation geometry. A representative example is the bounded-domain analysis of Koufos and Dettmann~\cite{Koufos2019}. Such methods complement the analytical approaches above by supporting reliable empirical estimation and more realistic stochastic-geometry modeling in finite deployment regions.

\subsection{Delaunay Simplices: Integral Forms, Moments, and Exact PDFs/CDFs}

Compared with the Voronoi case, PDTs admit substantially more tractable analytical progress. This is largely because Delaunay cells are simplices with a fixed combinatorial structure (consisting of $d+1$ vertices), whereas Voronoi cells are random polytopes with a random number of faces and complex geometric dependencies. Consequently, the typical simplex volume distribution can be characterized through integral representations, exact raw moments, and closed-form PDFs and CDFs.

\subsubsection{Integral Representations in Low Dimensions}

Early work established Mellin--Barnes-type integral representations for the Poisson--Delaunay simplex-volume distribution. Rathie~\cite{Rathie1992} developed a general integral framework and obtained explicit PDFs for \(d=1\) and \(d=2\). For \(d=3\), Muche~\cite{Muche1996} derived a triple-integral representation. These results are mathematically rigorous but can be cumbersome to evaluate numerically.

\subsubsection{Moment-Based Framework}
A major advantage of the Delaunay case is the availability of closed-form expressions for all raw moments, building on classical results by Miles~\cite{MILES197085}. This moment identity yields all raw moments, supports the determinacy of distributions, and offers a powerful tool for validating simulations and approximations. However, moments alone do not directly yield a closed-form PDF without additional inverse-transform techniques.

\subsubsection{Inverse Mellin Transform and Meijer's-$G$ Closed Forms}

A breakthrough was achieved by Xia and A\"{\i}ssa~\cite{8358978}, who resolved the long-standing absence of a unified and dimension-explicit characterization of the typical Poisson--Delaunay simplex volume distribution. Building on the exact raw-moment expressions, they applied inverse Mellin transforms to derive exact closed-form PDFs and CDFs in terms of Meijer's $G$-functions for arbitrary spatial dimensions. Importantly, the resulting framework retains the same functional form across dimensions, with the dimension dependence entering through the parameters of Meijer's $G$-function, and recovers the classical low-dimensional results for \(d=1\) and \(d=2\) after parameter specialization and simplification. This work therefore establishes an analytical bridge from exact moment identities to unified distributional representations.

\subsubsection{High-Dimensional and Asymptotic Analysis}
Recent work has also investigated the high-dimensional regime \(d\to\infty\), where random tessellations exhibit pronounced geometric concentration and admit sharp scaling laws. A representative contribution is due to Gusakova and Th\"ale~\cite{AHL_2021}, which characterizes the volume of simplices in high-dimensional PDTs. In this regime, simplex-volume functionals often concentrate around deterministic scales, reflecting measure-concentration effects in \(\mathbb{R}^d\). While valuable for theory and for connections to information theory, such results are primarily asymptotic and do not yield explicit finite-$d$ PDFs.

\subsection{Comparative Summary and Discussion}

Table~\ref{tab:comparison} summarizes representative methods for characterizing Poisson--Voronoi and Poisson--Delaunay typical-cell volume distributions, together with their principal advantages and limitations. For Poisson--Voronoi cells, simulation, parametric approximations, and moment-based methods remain the principal computationally tractable approaches. Exact planar integral formulae provide rigorous conditional representations but require high-dimensional integration and an infinite mixture over possible side counts. More recently, the exact scale--shape factorization has revealed a conditional Gamma structure and yielded exact mixture, transform, and moment representations in arbitrary dimensions. Nevertheless, because the normalized shape laws and the unbounded effective-facet-count mixture remain implicit, a tractable unconditional closed-form PDF or CDF is still unavailable.

\begin{table*}[!t]
\caption{Comparison of methods for characterizing typical-cell volume distributions in PVTs and PDTs.}
\label{tab:comparison}
\centering
\scriptsize
\renewcommand{\arraystretch}{1.12}
\setlength{\tabcolsep}{3.5pt}

\begin{tabularx}{\textwidth}{
@{}
>{\raggedright\arraybackslash}p{2.7cm}
>{\raggedright\arraybackslash}p{3.6cm}
Y
Y
@{}
}
\toprule
\textbf{Method Category}
& \textbf{Representative Works}
& \textbf{Advantages}
& \textbf{Limitations} \\
\midrule

\multicolumn{4}{c}{\textbf{Poisson--Voronoi Tessellation}} \\
\midrule

Simulation-based estimation
& Hinde and Miles~\cite{hinde1980monte}; Kumar \emph{et al.}~\cite{Kumar1992}; Lazar \emph{et al.}~\cite{Lazar2013}
& High-accuracy empirical PDFs, moments, and geometric statistics from Monte Carlo simulation
& No closed form; computationally costly in high dimensions \\
\addlinespace[2pt]

Parametric distribution fitting
& Kiang~\cite{kiang1966random}; Weaire \emph{et al.}~\cite{weaire1986distribution}; Tanemura~\cite{tanemura2003statistical}; Ferenc and N\'eda~\cite{ferenc2007size}
& Simple Gamma and generalized-Gamma approximations; widely used in engineering
& Empirical fitting; accuracy depends on parameterization, dimension, and tail behavior \\
\addlinespace[2pt]

Moment-based approaches
& Okabe \emph{et al.}~\cite{Okabe2000}; Hayen and Quine~\cite{hayen2002areas}; Muche and Ballani~\cite{MucheBallani2011}
& Exact or benchmark low-order moments; supports moment matching and parametric modeling
& Higher-order moments are difficult to obtain; does not directly yield the full PDF \\
\addlinespace[2pt]

Exact integral formulae
& Calka~\cite{Calka2003}
& Exact integral representations for the typical 2D cell area, perimeter, and related characteristics
& Multidimensional integrals; no tractable closed-form PDF/CDF \\
\addlinespace[2pt]

Exact scale--shape factorization
& Shi and Xia~\cite{xia2026scaleshape}
& Exact conditional Gamma scale; exact mixture, transform, and moment identities for arbitrary $d$
& Normalized shape-factor laws and mixing over the unbounded effective-facet count prevent a tractable unconditional PDF \\
\addlinespace[2pt]

Analytical and asymptotic results
& Zuyev~\cite{Zuyev1992}; M{\o}ller~\cite{Moller1994}; Calka~\cite{Calka_2002,Calka2005,calka2025voronoitail}; Hug \emph{et al.}~\cite{HugReitznerSchneider2004}; Schulte~\cite{SCHULTE2012285}; Irlbeck \emph{et al.}~\cite{irlbeck2025pvtd}
& Rigorous structural and asymptotic results on radii, large cells, extremes, geometric CLTs, and high-$d$ limits
& Generally do not yield tractable closed-form PDFs for the typical-cell volume \\
\addlinespace[2pt]

Finite-window analysis and boundary correction
& Koufos and Dettmann~\cite{Koufos2019}
& Mitigates finite-domain bias and complements tail asymptotics
& Tailored to bounded observation windows; no universal closed form \\
\midrule

\multicolumn{4}{c}{\textbf{Poisson--Delaunay Tessellation}} \\
\midrule

Integral representations
& Rathie~\cite{Rathie1992}; Muche~\cite{Muche1996}
& Rigorous integral representations, including explicit low-dimensional forms
& Cumbersome expressions; numerically challenging in higher dimensions \\
\addlinespace[2pt]

Moment-based framework
& Miles~\cite{MILES197085}
& Closed-form expressions for all raw moments; supports distribution determinacy
& Does not directly yield the PDF; requires inverse transforms \\
\addlinespace[2pt]

Inverse Mellin transform and Meijer's $G$-function
& Xia and A{\"i}ssa~\cite{8358978}
& Unified closed-form PDF/CDF for arbitrary $d$; supported by symbolic and numerical software
& Requires advanced special functions and careful numerical implementation \\
\addlinespace[2pt]

High-dimensional asymptotic analysis
& Gusakova and Th{\"a}le~\cite{AHL_2021}
& High-dimensional scaling laws and concentration results
& Asymptotic results only; no explicit finite-$d$ PDFs \\
\bottomrule
\end{tabularx}
\end{table*}

In contrast, Poisson--Delaunay simplices admit a more structured analytical treatment. Exact analytical expressions for their raw moments can be combined with inverse Mellin transforms to yield unified closed-form PDFs and CDFs in terms of Meijer's $G$-functions for arbitrary spatial dimensions. Meijer's $G$-function is a general special function defined by a Mellin--Barnes contour integral and encompasses many standard functions as special cases \cite{Mathai93}. Although the resulting formulas involve an advanced special function, they can be evaluated using modern symbolic and numerical software such as MATLAB, thereby providing rigorous and numerically accessible benchmarks for high-dimensional stochastic-geometry modeling and performance analysis. Their accurate numerical evaluation, however, requires careful implementation.

From an engineering perspective, these methods reflect different balances among analytical exactness, geometric interpretability, implementation complexity, and numerical efficiency. Gamma-type approximations offer compact expressions and low computational cost, making them tractable surrogates for many Poisson--Voronoi applications. Exact planar integral formulations preserve detailed geometric structure but require demanding high-dimensional integration. The scale--shape factorization provides an exact general-dimensional structural framework and a principled basis for mixture approximations, although its normalized shape laws and facet-count mixture remain implicit. For Poisson--Delaunay simplices, formulations based on Meijer's $G$-function provide exact distributional characterizations in a unified functional form, with numerical evaluation facilitated by computational tools. These distinct analytical frameworks motivate a closer examination of the scaling laws and structural properties governing typical-cell volume distributions, as developed in the next section.

Fig.~\ref{fig:analytical-pathways} further summarizes the complementary analytical pathways for Poisson--Voronoi and Poisson--Delaunay volume distributions and connects them to their representative wireless-network roles.

%
\section{Typical-Cell Volume Distributions in PVTs and PDTs}
\label{sec:PDFs_of_volumes}

This section reviews probabilistic characterizations of \emph{typical-cell} volume distributions in PVTs and their dual PDTs. We first recall the underlying spatial model and distinguish the Palm constructions of the two typical geometric objects. For Poisson--Voronoi cells, we present the recently developed exact scale--shape factorization, followed by Gamma-type approximations, benchmark moments, and intensity-scaling properties. For Poisson--Delaunay simplices, we review exact dimension-explicit PDF and moment formulas based on Mellin-transform techniques and Meijer's-$G$ representations, together with established high-dimensional scaling and limit results for the logarithmic simplex volume.

\subsection{Poisson Point Process Model and Typical Objects}

A stationary homogeneous PPP $\Phi$ with intensity $\lambda>0$ on $\mathbb{R}^d$ is a simple point process such that, for any Borel set $\mathcal A\subset\mathbb{R}^d$, the random count $N(\mathcal A)\triangleq|\Phi\cap\mathcal A|$ is Poisson distributed with mean $\lambda|\mathcal A|$, and counts in disjoint sets are independent, where $|\cdot|$ denotes the Lebesgue measure. Owing to its translation invariance and complete spatial randomness, the PPP is the canonical model in stochastic geometry and a standard baseline for large-scale wireless-network modeling.

Throughout this paper, ``typical'' cells and simplices are understood in the Palm sense of their corresponding stationary processes. Under the Palm distribution of the underlying PPP $\Phi$, the typical Poisson--Voronoi cell is the Voronoi cell generated by the point conditioned to lie at the origin~\cite{Okabe2000}. We denote this cell by $\mathcal{V}^{\ast}$ and its volume by $V\triangleq|\mathcal{V}^{\ast}|$. In contrast, the typical Poisson--Delaunay simplex is defined under the Palm distribution of the stationary particle process of Delaunay simplices, with the circumcenter used as the center function~\cite{AHL_2021}. Equivalently, this particle process may be represented as the point process of simplex circumcenters marked by the corresponding centered simplices. The typical simplex is obtained by selecting each simplex with equal weight and translating the selected simplex so that its circumcenter is at the origin. We denote its volume by $V_d$. This object is distinct from a simplex associated with a typical Poisson point and from the volume-weighted Delaunay simplex containing the origin.

Despite the geometric duality between Voronoi and Delaunay tessellations, the analytical tractability of their typical-cell volumes differs substantially. Poisson--Voronoi cells are random polytopes with random facet counts and complex geometric dependencies. Although the recently developed scale--shape factorization provides an exact structural and mixture representation in arbitrary dimensions, the normalized shape laws and the unbounded facet-count mixture remain implicit, preventing a tractable unconditional closed-form PDF or CDF. By contrast, Poisson--Delaunay cells are simplices with a fixed combinatorial structure, enabling exact raw-moment formulas and dimension-explicit PDFs and CDFs.

\begin{figure}[!t]
   	\centering
   	\includegraphics[width=0.475\textwidth]{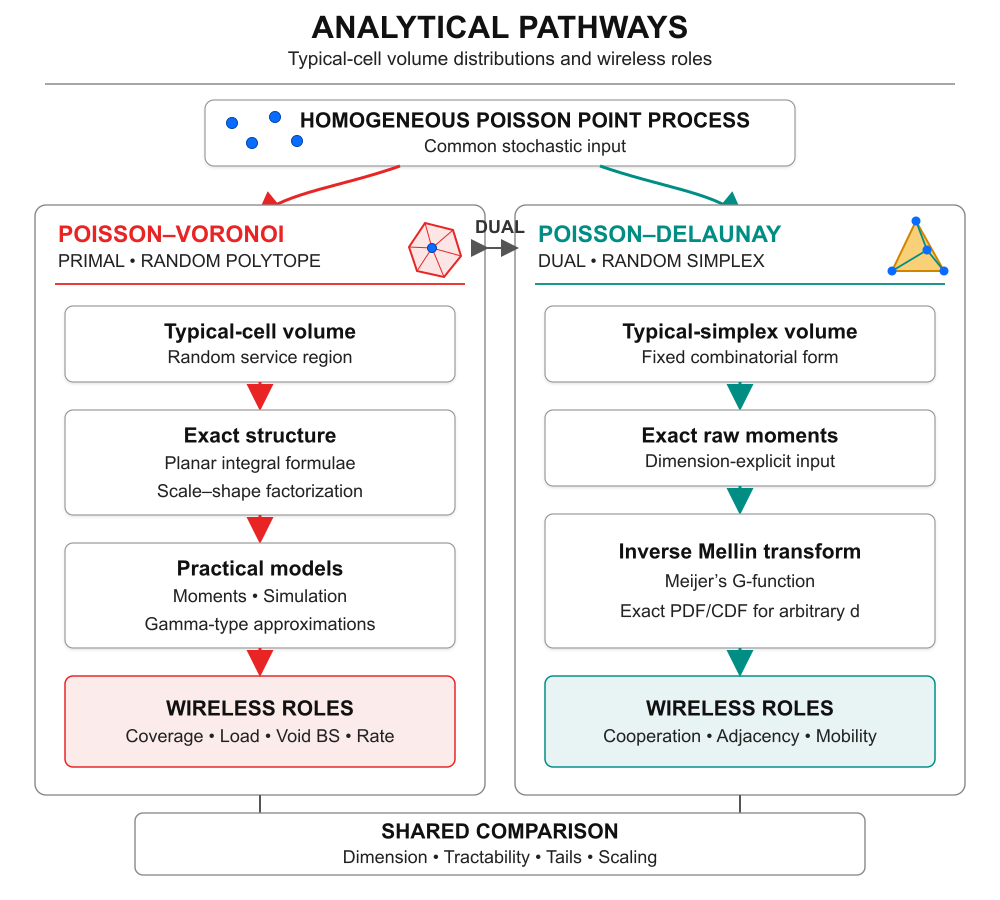}
	\caption{Analytical pathways and representative wireless-network roles of Poisson--Voronoi and Poisson--Delaunay typical-cell volume distributions. The Poisson--Voronoi pathway combines exact structural results with simulations, moments, and Gamma-type approximations, whereas the Poisson--Delaunay pathway derives exact dimension-explicit distributions from raw moments through inverse Mellin transforms and Meijer's $G$-function.}    
	\label{fig:analytical-pathways}
\end{figure}

\subsection{Poisson--Voronoi Typical-Cell Volume Distributions}
\label{subsec:voronoi}

Given a locally finite point set $\Phi=\{x_i\}\subset\mathbb{R}^d$, the Voronoi cell associated with $x_i$ is defined as
\begin{equation}
	\mathcal{V}_i
	= \Bigl\{ y\in\mathbb{R}^d : \|y-x_i\|\le \|y-x_j\|,\ \forall j\neq i\Bigr\}.
\label{Eq-PVT}
\end{equation}
The collection $\mathscr{V}=\{\mathcal{V}_i\}$ forms a partition of $\mathbb{R}^d$ into convex polytopes. When $\Phi$ is a homogeneous PPP, $\mathscr{V}$ is called a \emph{Poisson--Voronoi tessellation (PVT)}.

For $d=1$, the typical Poisson--Voronoi cell is an interval, and its length admits an explicit Gamma (Erlang) law~\cite{Okabe2000}. For $d=2$, exact integral formulae for the typical-cell area and perimeter are available~\cite{Calka2003}. More recently, an exact scale--shape factorization has revealed additional conditional structure in arbitrary dimensions~\cite{xia2026scaleshape}, as discussed below. Nevertheless, a tractable unconditional PDF of the typical-cell volume remains unavailable in general for $d\geq 2$. Practical modeling therefore continues to rely largely on Monte Carlo estimation, Gamma-type fitting, and low-order moment benchmarks.

\subsubsection{Exact Scale--Shape Factorization}
\label{subsubsec:scale_shape}

A recent scale--shape factorization provides an exact structural representation of the Palm-typical Poisson--Voronoi cell volume in arbitrary spatial dimensions~\cite{xia2026scaleshape}. Let \(V\) denote the cell volume, \(K\) its number of effective facets, and \(Q\) the volume of its Voronoi flower. The flower is the exclusion region whose interior must contain no additional Poisson nuclei for a candidate configuration of effective neighboring nuclei to generate the Palm cell. Since both \(V\) and \(Q\) are homogeneous of degree \(d\) in the neighboring-nucleus locations, their ratio is invariant under radial rescaling and depends only on the normalized configuration.

Conditional on \(K=k\), the flower volume \(Q\) follows a Gamma distribution with shape \(k\) and rate \(\lambda\). Let \(A_k=V/Q\) denote the scale-invariant cell-to-flower volume ratio, which captures the aspect of the normalized cell geometry relevant to its volume. Defining \(Y=\lambda V\) and \(Z_k=\lambda Q\), the factorization can then be written as
\begin{equation}
\label{eq:PV_scale_shape}
Y\mid\{K=k\}
\overset{\mathrm d}=A_kZ_k,
\qquad
Z_k\sim\operatorname{Gamma}(k,1).
\end{equation}
Here, \(\operatorname{Gamma}(k,1)\) denotes the Gamma distribution with shape \(k\) and unit rate. Under the conditional law given \(K=k\), \(A_k\) and \(Z_k\) are independent, and \(A_k\) satisfies
\begin{equation}
\label{eq:PV_shape_support}
0<A_k\leq 2^{-d}.
\end{equation}
This support bound follows from the universal geometric inequality \(Q\geq 2^dV\).

Let \(p_{d,k}=\mathbb{P}(K=k)\), and, for each \(k\) with \(p_{d,k}>0\), let \(\eta_{d,k}\) denote the distribution of \(A_k\). Mixing \eqref{eq:PV_scale_shape} over the shape variable and the effective-facet count yields, for \(v>0\),
\begin{align}
F_V(v)
&=
\sum_{k=d+1}^{\infty}p_{d,k}
\int_{0}^{2^{-d}}
\frac{\gamma(k,\lambda v/a)}{\Gamma(k)}
\,\eta_{d,k}(\mathrm da),
\label{eq:PV_scale_shape_cdf}\\
f_V(v)
&=
\sum_{k=d+1}^{\infty}p_{d,k}
\int_{0}^{2^{-d}}
\frac{\lambda^k v^{k-1}}{\Gamma(k)a^k}
\exp\!\left(-\frac{\lambda v}{a}\right)
\,\eta_{d,k}(\mathrm da),
\label{eq:PV_scale_shape_pdf}
\end{align}
where \(\gamma(k,x)\) is the lower incomplete Gamma function, and terms with \(p_{d,k}=0\) are interpreted as zero. Equivalently, conditional on \(K=k\) and \(A_k=a\), \(V\) has a Gamma density with shape \(k\) and rate \(\lambda/a\). The unconditional law is therefore an infinite mixture over both the normalized shape ratio and the effective-facet count.

The factorization explains the effectiveness and limitations of Gamma approximations. Two sources of variability govern possible departures from a single Gamma model: fluctuations of \(A_k\) within each facet-count stratum and mixing across the random effective-facet count. It also yields exact transform and moment identities, thereby connecting the underlying geometric structure to moment-based approximation methods, and provides a configuration-space framework for studying the lower tail. In one dimension, the factorization recovers the exact distribution because the typical cell has two effective facets and a constant cell-to-flower ratio. In the planar case, the normalized configurations admit explicit coordinates that connect the general factorization to classical integral representations. Nevertheless, the shape laws \(\eta_{d,k}\) and the unbounded facet-count mixture remain implicit. Consequently, the result provides exact structural and distributional representations, but not a tractable unconditional closed-form PDF or CDF.

\subsubsection{Gamma-Type Approximations}
\label{subsubsec:voronoi_gamma}

We next consider the two-parameter Gamma (TG) model. A standard approximation models $V$ by a Gamma distribution with shape $a>0$ and rate $b>0$:
\begin{equation}\label{eq:two_para}
	f_{\mathrm{TG}}(v) = \frac{b^a}{\Gamma(a)} v^{a-1} e^{-b v}, \quad v>0.
\end{equation}
Representative empirical fits reported in the literature include:
\begin{itemize}
\item 2D: Weaire \emph{et al.}~\cite{weaire1986distribution} ($a=b=3.61$) and DiCenzo and Wertheim~\cite{dicenzo1989monte} ($a=3.61$, $b=3.57$).
\item 3D: Kiang~\cite{kiang1966random} ($a=b=6$).
\end{itemize}


These parameter values apply to the intensity-normalized volume \(\widetilde V=\lambda V=V/\mathbb{E}[V]\), which is dimensionless and has unit mean because \(\mathbb{E}[V]=1/\lambda\). Equivalently, one may rescale the underlying PPP to unit intensity and express cell volumes in the corresponding unit-intensity volume units.

Then, we move to the three-parameter generalized Gamma (GG) model.
To increase flexibility, especially in fitting the tail behavior, the generalized Gamma family introduces an exponent parameter $c>0$:
\begin{equation}\label{eq:three_para}
	f_{\mathrm{GG}}(v)
	= c \frac{b^{a/c}}{\Gamma(a/c)} v^{a-1} \exp(-b v^c), \quad v>0,
\end{equation}
where $a$ controls the near-origin behavior, $b$ sets the scale, and $c$ modulates tail thickness (with $c=1$ reducing to the TG model).
Representative empirical fits include:
\begin{itemize}
\item 2D: Hinde and Miles~\cite{hinde1980monte} ($a=3.3095$, $b=3.0328$, $c=1.0787$) and
Tanemura~\cite{tanemura2003statistical} ($a=3.315$, $b=3.04011$, $c=1.078$).
\item 3D: Tanemura~\cite{tanemura2003statistical} ($a=4.8065$, $b=4.06342$, $c=1.16391$).
\end{itemize}

Finally, we address the dimension-dependent one-parameter baseline.
Ferenc and N\'eda~\cite{ferenc2007size} proposed a compact dimension-dependent Gamma-type approximation
\begin{equation}\label{eq:ferenc_neda}
	f_d(v)
	= \frac{\left(\frac{3d+1}{2}\right)^{(3d+1)/2}}{\Gamma\!\left(\frac{3d+1}{2}\right)}
	v^{(3d-1)/2} \exp\!\left(-\frac{3d+1}{2} v\right),
\end{equation}
which effectively fixes the parameters as functions of $d$ and provides a parameter-free reference for the \emph{normalized} volume (unit mean) in $d=1,2,3$.

\subsubsection{Moment Characterizations}
\label{subsubsec:voronoi_moments}

For the TG approximation \eqref{eq:two_para}, the first three raw moments are
\begin{align}
\mathbb{E}_{\mathrm{TG}}[V] &= \frac{a}{b}, \label{eq:gamma_m1}\\
\mathbb{E}_{\mathrm{TG}}[V^2] &= \frac{a(a+1)}{b^2}, \label{eq:gamma_m2}\\
\mathbb{E}_{\mathrm{TG}}[V^3] &= \frac{a(a+1)(a+2)}{b^3}. \label{eq:gamma_m3}
\end{align}
For the GG approximation \eqref{eq:three_para}, the corresponding raw moments are
\begin{align}
\mathbb{E}_{\mathrm{GG}}[V] &= \frac{b^{-1/c}\Gamma\!\left(\frac{a+1}{c}\right)}{\Gamma(a/c)}, \label{eq:ggamma_m1}\\
\mathbb{E}_{\mathrm{GG}}[V^2] &= \frac{b^{-2/c}\Gamma\!\left(\frac{a+2}{c}\right)}{\Gamma(a/c)}, \label{eq:ggamma_m2}\\
\mathbb{E}_{\mathrm{GG}}[V^3] &= \frac{b^{-3/c}\Gamma\!\left(\frac{a+3}{c}\right)}{\Gamma(a/c)}. \label{eq:ggamma_m3}
\end{align}
These formulas support moment matching, in which model parameters are calibrated using a small set of benchmark moments, and provide an interpretable basis for comparing the accuracy of different empirical fits.

Commonly used low-order moments (under $\lambda=1$) include:
\begin{itemize}
	\item 2D: $\mathbb{E}[V]=1$, $\mathbb{E}[V^2]\approx 1.2801$, $\mathbb{E}[V^3]\approx 1.9992$~\cite{hayen2002areas}.
	\item 3D: $\mathbb{E}[V]=1$, $\mathbb{E}[V^2]\approx 1.1790$~\cite{MucheBallani2011}, $\mathbb{E}[V^3]\approx 1.6230$\footnote{Estimated from the large-sample 3D Monte Carlo data in \cite{tanemura2003statistical}.}.
\end{itemize}

\subsubsection{Scaling With PPP Intensity}
\label{subsubsec:voronoi_scaling}
If $\Phi$ has intensity $\lambda \ne 1$, then scaling the point process by the factor $\lambda^{1/d}$ maps it to a unit-intensity PPP. Since volume scales as the power $d$ of length, the typical-cell volume satisfies
\begin{equation}
V \stackrel{\mathrm{d}}{=} \lambda^{-1} V_{(1)},
\label{eq:voronoi_scaling_dist}
\end{equation}
where $V_{(1)}$ denotes the typical-cell volume under the intensity $\lambda=1$.

Equivalently, for any moment order $k>0$,
\begin{equation}
\mathbb{E}[V^k] = \lambda^{-k}\,\mathbb{E}\!\left[V_{(1)}^{k}\right].
\label{eq:voronoi_scaling_mom}
\end{equation}
This is the precise basis for the common practice of reporting fitted parameters and moments under $\lambda=1$ and then rescaling to arbitrary intensities.

In practice, moment matching is most effective as a \emph{calibration/validation} tool rather than as a claim of exactness: it allows one to (i) fit TG/GG parameters from $(\mathbb{E}[V],\mathbb{E}[V^2],\mathbb{E}[V^3])$, and (ii) quantify how well different parametric families reproduce benchmark moments and simulated tails. From a theoretical standpoint, the Carleman condition provides a classical sufficient criterion for moment determinacy~\cite{akhiezer2020classical}, which conceptually supports moment-based approximation strategies under suitable regularity.

\begin{remark}[Exact and approximate Voronoi models]
The scale--shape factorization provides an exact structural representation rather than a directly evaluable parametric model. By identifying a conditional Gamma scale and the additional variability arising from normalized shapes and random facet counts, it explains both the recurring effectiveness and the limitations of Gamma-type approximations. Among these approximations, the TG model provides the simplest parametric description, whereas the GG model introduces an additional degree of freedom that generally improves fitting flexibility and tail fidelity. The Ferenc--N\'eda model instead provides a compact dimension-dependent baseline without explicit parameter estimation. The exact factorization and all three approximations are consistent with the PPP intensity-scaling property in \eqref{eq:voronoi_scaling_dist}--\eqref{eq:voronoi_scaling_mom}.
\end{remark}

\subsection{Poisson--Delaunay Simplex-Volume Distributions}
\label{subsec:delaunay}

The geometric dual of the PVT is the PDT. Given a locally finite point set $\Phi=\{x_i\}\subset\mathbb{R}^d$, a $d$-simplex
\begin{equation}
\label{Eq-PDT}
	\sigma = \mathrm{conv}\{x_0,\cdots,x_d\},
\end{equation}
with vertices $\{x_0,\cdots,x_d\}\subset\Phi$ is called a \emph{Delaunay simplex} if its circumsphere contains no other point of $\Phi$ in its interior. The collection of all such simplices forms a tessellation of $\mathbb{R}^d$.
Unlike PVTs, whose typical cells may have different numbers of faces, the fundamental elements of a PDT are uniformly $d$-simplices (triangles in 2D and tetrahedra in 3D).

Compared with the Poisson--Voronoi case, the simplex-volume distribution admits substantially more explicit analytical characterizations. In particular, for the typical $d$-dimensional Delaunay simplex, both the PDF and the complete sequence of raw moments can be expressed in closed form via Mellin transforms, leading to unified Meijer-$G$ representations that hold for all $d \ge 1$~\cite{8358978}. This exact characterization provides a rigorous benchmark for simulation and for assessing approximation accuracy. Also, it facilitates geometry-consistent constructions in practical applications, such as Delaunay-based cooperative transmission in 2D wireless networks~\cite{8976426} and in 3D scenarios~\cite{9151343}.

\subsubsection{Unified Meijer's-$G$ PDF}
\begin{figure*}[!t]
\begin{equation}
f_{V_{d}}(x)
= \frac{A_d}{x} \, \mathrm{G}_{p,\, q}^{m,\, 0}\left[ {B_d} \, x^2 \left \vert \begin{gathered}
\underbrace{\frac{d}{2}+\frac{1}{d}, \, \frac{d}{2}+\frac{2}{d}, \, \cdots, \, \frac{d}{2}+\frac{d-1}{d}}_{(d-1)\ \text{terms}}, \,
\underbrace{\frac{d+1}{2}, \, \frac{d+1}{2}, \, \cdots, \, \frac{d+1}{2}}_{(d-1)\ \text{terms}}
\\
\underbrace{1, \, \frac{3}{2}, \, \cdots, \, \frac{d}{2}}_{(d-1)\ \text{terms}}, \,
\underbrace{\frac{d^2+1}{2(d+1)}, \, \frac{d^2+3}{2(d+1)}, \, \cdots, \, \frac{d^2+2d+1}{2(d+1)}}_{(d+1)\ \text{terms}}
 \end{gathered}\right.\right].
\label{eq:delaunay_pdf}
\end{equation}
\end{figure*}

Let $V_d$ denote the volume of the typical $d$-dimensional Poisson--Delaunay simplex. For a homogeneous PPP with unit intensity, its PDF is given by \eqref{eq:delaunay_pdf}, where $m=2d$, $p=2d-2$, and $q=2d$. The representation is unified in the sense that the same functional form applies to every spatial dimension $d$, with the dimension dependence entering through the parameters of Meijer's $G$-function. The coefficients $A_d$ and $B_d$, which depend only on the dimension $d$, are given in \cite{8358978} as
\begin{align}
	A_d & = \frac{2^{d-\frac{1}{2}} (d+1)^{d^{2}/2} \, \Gamma\!\left(\frac{d^2}{2}\right)\Gamma^{d}\!\left(\frac{d+1}{2}\right)}
                {\pi \, d^{(d^2-1)/2} \, \Gamma(d)\, \Gamma\!\left(\frac{d^2+1}{2}\right) \prod\limits_{i=2}^{d}{\Gamma\!\left(\frac{i}{2}\right)}},  \label{eq:coef_A} \\
	B_d & = \left[\frac{2^{d-1} \pi^{(d-1)/2} d^{d/2} \, \Gamma\!\left(\frac{d+1}{2}\right)}{(d+1)^{(d+1)/2}}\right]^{2}.  \label{eq:coef_B}
\end{align}
Despite involving advanced special functions, these expressions are numerically implementable using standard symbolic/numerical tools, and therefore provide an \emph{exact} reference distribution across dimensions.

\subsubsection{Scaling With PPP Intensity}
\label{sec:delaunay_scaling}
Similar to \eqref{eq:voronoi_scaling_dist}, the scaling invariance of the PPP implies that $V_d$ scales inversely with the point intensity:
\begin{equation}
V_d(\lambda)\ \stackrel{\mathrm{d}}{=}\ \lambda^{-1} V_d(1),
\label{eq:delaunay_scaling_dist}
\end{equation}
which further implies that the PDF satisfies
\begin{equation}
	f_{V_d}(x;\lambda) = \lambda\, f_{V_d}(\lambda x;1),
\label{eq:pdf_scaling}
\end{equation}
while all raw moments scale as $\mathbb{E}[V_d^k]\propto \lambda^{-k}$.

\subsubsection{Closed Forms in Low Dimensions}
Substituting $d=1,2,3$ into \eqref{eq:delaunay_pdf} yields simplified closed forms~\cite{8358978} for all $x>0$:
\begin{itemize}
	\item $d = 1$:
	\begin{equation}
		f_{V_1}(x)=e^{-x}.
	\label{Eq.CaseI-PDF}
	\end{equation}

	\item $d = 2$:
	\begin{equation}
		f_{V_2}(x)
		=  \frac{8}{9}\pi x\,
		K_{\frac{1}{6}}^{2}\!\left(\frac{2\pi x}{3\sqrt{3}}\right).
	\label{Eq.CaseII-PDF}
	\end{equation}

	\item $d = 3$:
	\begin{equation} \hspace{-1em}
		f_{V_3}(x)
		= \frac{560\sqrt{2}}{81\pi x}\,
		\mathrm{G}_{3, \, 5}^{5, \, 0}\!\left(
		\frac{27\pi^2}{16}x^2 \;\middle|\;
		\begin{matrix}
		\frac{11}{6},\, \frac{13}{6},\, 2
		\\[2pt]
		1,\, \frac{3}{2},\, \frac{5}{4},\, \frac{3}{2},\, \frac{7}{4}
		\end{matrix}
		\right).
	\label{Eq.CaseIII-PDF}
	\end{equation}
\end{itemize}
The special cases \eqref{Eq.CaseI-PDF} and \eqref{Eq.CaseII-PDF} are also consistent with earlier low-dimensional representations (e.g., Mellin--Barnes integral forms). Moreover, numerical evaluations of \eqref{Eq.CaseIII-PDF} agree with those obtained from the complex triple-integral representation reported in \cite{Muche1996}. 

\begin{figure*}[!t]
\begin{equation}
\mathbb{E}\!\left[V_d^k\right]
= \frac{\Gamma(d+k) \, \Gamma\!\left(\frac{d^2}{2}\right)\Gamma\!\left(\frac{d^2+dk+k+1}{2}\right)\Gamma^{d-k+1}\!\left(\frac{d+1}{2}\right)}
{(2^{d}\,\pi^{(d-1)/2}\lambda)^k \, \Gamma(d) \, \Gamma\!\left(\frac{d^2+1}{2}\right)\Gamma\!\left(\frac{d^2+dk}{2}\right)\Gamma^{d+1}\!\left(\frac{d+k+1}{2}\right)}
\,\prod_{i=2}^{d+1}{\frac{\Gamma\!\left(\frac{k+i}{2}\right)}{\Gamma\!\left(\frac{i}{2}\right)}}, \quad k=1,2,3,\ldots
\label{eq:delaunay_moment}
\end{equation}
\end{figure*}

\subsubsection{Exact Moments and High-Dimensional Log-Volume Scaling Law}
\label{Section-AsymptoticAnalysis}

For fixed $\lambda>0$, the $k$-th raw moment of the volume of the typical $d$-dimensional Poisson--Delaunay simplex ($V_d$) admits the exact closed-form expression in~\eqref{eq:delaunay_moment}; see~\cite{8358978}. This dimension-explicit identity serves as a rigorous benchmark for numerical validation and provides a basis for deriving high-dimensional scaling laws.

Fix $k\in\mathbb{N}_{+}$ independently of $d$. Applying Stirling's formula to the exact moment identity~\eqref{eq:delaunay_moment} and collecting the dominant terms shows that the $\Theta(d^2\log d)$ contributions cancel, yielding
\begin{equation}
\label{eq:delaunay_moment_scaling_final}
\log \mathbb{E}[V_d^k]
= -\frac{k}{2}\, d\log d + \mathcal{O}(d) - k\log\lambda,
\quad d\to\infty .
\end{equation}

In particular, for $k=1$, \eqref{eq:delaunay_moment_scaling_final} implies
\begin{equation}
\label{eq:delaunay_typical_scale_final}
\mathbb{E}[V_d]
= \lambda^{-1}\, d^{-d/2} e^{\mathcal{O}(d)},
\quad d\to\infty .
\end{equation}
Thus, at leading logarithmic order, the first-moment scale of the typical simplex volume is governed by $d^{-d/2}$, revealing the dominant $d\log d$ exponent underlying its high-dimensional contraction.

Moment asymptotics such as \eqref{eq:delaunay_moment_scaling_final} describe only the growth rate of expectations and do not capture the concentration properties of $\log V_d$. For the typical Poisson--Delaunay simplex, Gusakova and Th{\"a}le \cite{AHL_2021} established a precise high-dimensional limit theory for the logarithmic volume:
\begin{align}
\mathbb{E}[\log V_d]
&=
-\frac{d}{2}\log d-\log\lambda+\mathcal{O}(d), \label{eq:delaunay_typical_scale_ahl2021}\\
\operatorname{Var}(\log V_d)
&=
\frac{1}{2}\log d+\mathcal{O}(1), \label{eq:delaunay_typical_scale_ahl2021var}
\end{align}
together with a central limit theorem and a moderate deviation principle for the standardized log-volume. Intuitively, these relations show that the logarithmic volume is dominated by the term $-\tfrac{d}{2}\log d$, corresponding at leading logarithmic order to a volume scale of $d^{-d/2}$, whereas its random fluctuations increase at a much slower rate.

Specializing the results of Gusakova and Th\"ale~\cite{AHL_2021} to the typical Poisson--Delaunay simplex and reformulating them in terms of a dimension- and intensity-normalized logarithmic volume makes the dominant high-dimensional scaling explicit. Motivated by the leading-order behavior in \eqref{eq:delaunay_typical_scale_ahl2021}, we define
\begin{equation}
\label{eq:def_Zd}
Z_d \triangleq \frac{-\log(\lambda V_d)}{d\log d}.
\end{equation}
By the exact intensity-scaling property of a homogeneous PPP, the distribution of $\lambda V_d$ is independent of $\lambda$. The normalization in \eqref{eq:def_Zd} therefore isolates the principal high-dimensional contraction of the simplex volume. The following corollary summarizes the resulting limiting behavior and fluctuations of $Z_d$.

\begin{corollary}[High-Dimensional Limit Theory]
\label{corollary:Zd_limit_theory}
Let $V_d$ denote the volume of the typical $d$-dimensional Poisson--Delaunay simplex generated by a homogeneous PPP of fixed intensity $\lambda>0$, and let $Z_d$ be defined by \eqref{eq:def_Zd}. Then the following asymptotic properties hold as $d\to\infty$.
\begin{enumerate}
\item Law of Large Numbers (LLN):
\begin{equation}
\label{eq:Zd_lln}
Z_d \xrightarrow{P} \tfrac12 .
\end{equation}

\item Central Limit Theorem (CLT):
\begin{equation}
\label{eq:Zd_clt}
\sqrt{2\log d}\,d\bigl(Z_d-\mathbb{E}[Z_d]\bigr)
\xrightarrow{\mathcal{D}} \mathcal{N}(0,1).
\end{equation}

\item Moderate Deviation Principle (MDP):
For any sequence $a_d\to\infty$ satisfying $a_d=o(\sqrt{\log d})$, the sequence
\[
\frac{Z_d-\mathbb{E}[Z_d]}
{a_d\sqrt{\operatorname{Var}(Z_d)}}
\]
obeys a moderate deviation principle with speed $a_d^2$ and quadratic rate function $I(x)=x^2/2$.

\end{enumerate}
\end{corollary}

\begin{IEEEproof}
See the Appendix.
\end{IEEEproof}

Corollary~\ref{corollary:Zd_limit_theory} brings the high-dimensional concentration and fluctuation results of Gusakova and Th\"ale~\cite{AHL_2021} into a unified normalization for the typical Poisson--Delaunay simplex volume. In particular, it shows that $Z_d$ converges in probability to $1/2$ and that
\[
\sqrt{\operatorname{Var}(Z_d)}
\sim \frac{1}{d\sqrt{2\log d}}.
\]
Thus, the normalized logarithmic volume becomes increasingly concentrated around its deterministic limit. Equivalently, after removal of the intensity factor, the dominant volume scale at leading logarithmic order is $d^{-d/2}$, whereas the random fluctuations of the logarithmic volume occur on the scale $\sqrt{\log d}$, which is negligible relative to the leading $d\log d$ term.

Here, $d$ denotes the ambient spatial dimension of the geometric model and is not a measure of network density, node population, or cooperation level. The limit $d\to\infty$ should therefore be interpreted as a mathematical asymptotic regime that reveals dimension-dependent geometric structure, rather than as a direct physical limit of a wireless network. From an engineering perspective, the resulting concentration phenomena may provide asymptotic benchmarks for abstract high-dimensional models, distributed-learning architectures, and geometric embeddings. They do not, however, imply corresponding performance trends in increasingly dense or highly cooperative networks. Their relevance to practical two- and three-dimensional deployments must be assessed separately through finite-dimensional analysis.

\begin{remark}[High-Dimensional Asymmetry]
Corollary~\ref{corollary:Zd_limit_theory} shows that the typical Poisson--Delaunay simplex has the volume scale $d^{-d/2}$ at leading logarithmic order. By contrast, established volume-degeneracy results for the typical Poisson--Voronoi cell show that, for fixed intensity $\lambda$, its volume concentrates around its mean $1/\lambda$ as $d\to\infty$~\cite{Alishahi_Sharifitabar_2008,irlbeck2025pvtd}. If $V_d^{\mathrm{PV}}$ denotes its volume, then
\[
\lambda V_d^{\mathrm{PV}}\xrightarrow{P}1
\quad\Longrightarrow\quad
\frac{-\log(\lambda V_d^{\mathrm{PV}})}{d\log d}
\xrightarrow{P}0.
\]

Thus, despite their geometric duality, the two tessellations exhibit different high-dimensional volume scales. The intensity-normalized volume of the typical Poisson--Voronoi cell converges in probability to one, whereas the typical Poisson--Delaunay simplex contracts at the leading logarithmic rate $d^{-d/2}$.

From a modeling perspective, this distinction suggests different high-dimensional approximations for the two structures. Voronoi volumes may be represented by their asymptotically stable intensity-normalized scale, whereas Delaunay simplex volumes require an explicitly dimension-dependent contraction factor. Such approximations may be useful in abstract high-dimensional models of coverage regions, adjacency structures, and multi-node interactions, although their accuracy at finite dimensions must be assessed separately.
\end{remark}

%
\section{Comparative Numerical Analysis}
\label{sec:results}

Building on the analytical characterizations in Section~\ref{sec:PDFs_of_volumes}, we present a comparative numerical study of typical-cell and simplex-volume distributions in PVTs and PDTs. Our objectives are twofold: (i) to extract quantitative, dimension-dependent trends under a unified simulation and normalization framework, and (ii) to assess the accuracy of commonly used tractable approximations. Specifically, we compare low-order raw moments, goodness-of-fit metrics for Gamma-type models of Poisson--Voronoi volumes, and empirical PDFs with particular emphasis on the tail regime. These statistics are directly relevant for stochastic-geometry-based wireless-network analysis, where Voronoi volumes govern coverage/load regions and cell-edge distances, while Delaunay simplices underpin adjacency-based cooperation regions and geometry-consistent clustering.

\subsection{Moment Comparison and Geometric Interpretation}

Table~\ref{tab:moment_summary} reports the first three raw moments of typical Voronoi-cell volumes and typical Delaunay-simplex volumes for $d=1,2,3$. Consistent with the PPP scaling invariance discussed in Sections~\ref{subsubsec:voronoi_scaling} and~\ref{sec:delaunay_scaling}, the $k$-th moment scales as $\lambda^{-k}$ for both tessellations. This scaling serves as a convenient sanity check for the simulation pipeline and the mean normalization used in subsequent fitting comparisons.

In one dimension ($d=1$), both typical-object volumes admit closed-form expressions for their moments. Specifically, for an intensity-$\lambda$ PPP on $\mathbb{R}$, the typical Voronoi cell length equals one half of the sum of two i.i.d. \ exponential random variables with mean $1/\lambda$, whereas the length of a typical Delaunay edge is exponential with mean $1/\lambda$. Consequently,
\begin{equation}
	\mathbb{E}_{\mathrm{Vor}}[V^k]= \frac{(k+1)!}{(2\lambda)^{k}}, \quad
	\mathbb{E}_{\mathrm{Del}}[V^k]=\frac{k!}{\lambda^k},
\end{equation}
which provide exact reference values for validating Monte Carlo estimates in higher dimensions.

\begin{table}[!t]
	\caption{Comparison of the first three raw moments of typical-cell volumes in PVTs and PDTs across dimensions.}
	\centering
	\scriptsize
	\renewcommand{\arraystretch}{1.5}
	\begin{tabular}{ccccc}
		\toprule
		 & $d$  & 1 & 2 & 3\\
		\midrule
		& First moment & $\frac{1}{\lambda}$ & $\frac{1}{\lambda}$ & $\frac{1}{\lambda}$\\
	 Voronoi  	& Second moment & $\frac{1.5}{\lambda^2}$ & $\frac{1.2801}{\lambda^2}$ & $\frac{1.1790}{\lambda^2}$\\
		& Third moment & $\frac{3}{\lambda^3}$ & $\frac{1.9992}{\lambda^3}$ & $\frac{1.6230}{\lambda^3}$\\
		\midrule
		& First moment & $\frac{1}{\lambda}$ & $\frac{0.5}{\lambda}$ & $\frac{0.1478}{\lambda}$\\
	 Delaunay  	& Second moment & $\frac{2}{\lambda^2}$ & $\frac{0.4433}{\lambda^2}$ & $\frac{0.0372}{\lambda^2}$\\
		& Third moment & $\frac{6}{\lambda^3}$ & $\frac{0.5699}{\lambda^3}$ & $\frac{0.0135}{\lambda^3}$\\
		\bottomrule
	\end{tabular}
\label{tab:moment_summary}
\end{table}

Beyond verifying the scaling laws, Table~\ref{tab:moment_summary} reveals a fundamental geometric contrast between Voronoi and Delaunay structures as the dimension increases. By construction, the typical Voronoi-cell volume satisfies $\mathbb{E}[V]=1/\lambda$ for all $d$, reflecting the space-partition property: each point of the PPP ``owns'' an average region of volume $1/\lambda$. In contrast, the mean Poisson--Delaunay simplex volume decreases rapidly with $d$, and its higher-order moments decay even faster. More precisely, at leading logarithmic order, the mean volume shrinks at the rate $d^{-d/2}$, as established in the previous section. This contrast reflects the dual nature of the two tessellations: while Voronoi cells preserve a dimension-invariant average volume, the increasing combinatorial complexity of the Delaunay tessellation is accompanied by a sharp contraction of individual simplices as the dimension grows.

This discrepancy is not a numerical artifact but rather a direct consequence of the empty-circumsphere constraint that defines Delaunay simplices. By \eqref{Eq-PDT}, a Delaunay simplex is a \emph{local proximity object} formed by $d+1$ neighboring PPP points whose circumsphere contains no other points. As $d$ increases, measure-concentration effects in $\mathbb{R}^d$ render neighborhoods increasingly ``thin'' and impose stronger restrictions on feasible empty circumspheres. Consequently, typical Delaunay simplices become progressively more localized, and their volumes shrink super-polynomially fast with $d$, consistent with the asymptotic scaling in \eqref{eq:Zd_clt}.

A complementary indicator is the \emph{relative variability} of cell volumes. Define the squared coefficient of variation (CV) as
\begin{equation}
\mathrm{CV}^2 \triangleq \frac{\mathrm{Var}(V)}{\mathbb{E}[V]^2}
= \frac{\mathbb{E}[V^2]}{\mathbb{E}[V]^2}-1.
\end{equation}
From Table~\ref{tab:moment_summary}, Voronoi cells exhibit $\mathrm{CV}^2\approx 0.2801$ in 2D and $\mathrm{CV}^2\approx 0.1790$ in 3D, indicating moderate dispersion around the mean. For Delaunay simplices, the corresponding values are significantly larger (e.g., $\mathrm{CV}^2\approx 0.7732$ in 2D and $\mathrm{CV}^2\approx 0.7025$ in 3D). This shows that although typical simplex volumes decrease rapidly with dimension, their \emph{relative} fluctuations remain substantial in low dimensions. In higher dimensions, the moment-level asymptotic scaling in \eqref{eq:delaunay_moment_scaling_final} suggests increasing concentration of the log-volume. However, such concentration effects are not yet dominant for $d\le3$, where finite-dimensional fluctuations remain significant relative to the asymptotic scaling.

From a wireless-network perspective, these geometric trends admit a natural interpretation. Voronoi volumes characterize \emph{coverage or load regions}, and their tails govern load imbalance and cell-edge performance~\cite{6042301}. Delaunay simplices, in contrast, often define \emph{adjacency-based cooperation clusters}, e.g., coordinated multi-point (CoMP) groups in cooperative wireless transmission~\cite{8358978,9151343}. The rapid shrinkage of typical Delaunay volumes with increasing dimension therefore implies that such cooperation regions become increasingly compact, favoring localized multi-point association while inherently limiting the spatial footprint of coordination. These geometric observations are consistent with the high-dimensional log-volume limit theory developed earlier: the normalization $Z_d = -\log(\lambda V_d)/(d\log d)$ extracts the universal contraction scale of Delaunay simplices, whose deterministic limit $1/2$ quantifies the balance between combinatorial growth and volumetric shrinkage in high-dimensional PDTs.

\subsection{Quantitative Evaluation of Gamma-Type Approximations for Voronoi Volumes}

We next assess the accuracy of the Gamma-type approximations introduced in Section~\ref{subsubsec:voronoi_gamma}, using Monte Carlo PDFs as ground truth.

To obtain empirical benchmarks, we generate stationary PPP realizations in a bounded observation window $\mathcal{W} \subset\mathbb {R}^d$ with intensity $\lambda$ and construct the associated Voronoi tessellations using standard computational-geometry routines. To mitigate boundary effects, which otherwise bias typical-cell statistics by truncating cells near the boundary $\partial\mathcal W$, we adopt a guard-zone strategy: tessellations are computed over an expanded window $\mathcal W_{\rm ext}\supset\mathcal W$, while statistics are collected only for cells whose nuclei lie in $\mathcal W$ and whose associated cells are fully contained in $\mathcal W_{\rm ext}$. This yields an approximately unbiased empirical sample of typical-cell volumes. All simulations are performed under $\lambda=1$ without loss of generality, and results for general $\lambda$ follow directly from the scaling law in \eqref{eq:voronoi_scaling_dist}.

For each dimension ($d=2,3$), we generate a large number of typical-cell samples and estimate the empirical PDF using a histogram or kernel density estimation (KDE) after mean normalization. Candidate fitted models are compared using three complementary metrics:
\begin{itemize}
\item Mean squared error (MSE): $\mathbb{E}[(\hat f(v)-f(v))^2]$, which penalizes large local deviations;
\item Mean absolute error (MAE): $\mathbb{E}[|\hat f(v)-f(v)|]$, which is more robust to isolated spikes;
\item $R^2$: the coefficient of determination, computed over a discretized evaluation grid $\{v_i\}_{i=1}^{N}$ as
\begin{equation}
R^2 \triangleq 1-\frac{\sum_{i=1}^{N}\left(f_\mathrm{emp}(v_i)-f_\mathrm{fit}(v_i)\right)^2}{\sum_{i=1}^{N}\left(f_\mathrm{emp}(v_i)-\bar f_\mathrm{emp}\right)^2}, 
\end{equation}
with
\begin{equation}
\bar f_\mathrm{emp} \triangleq \frac{1}{N}\sum_{i=1}^{N} f_\mathrm{emp}(v_i),
\end{equation}
where $f_\mathrm{emp}(v)$ denotes the empirical PDF and $f_\mathrm{fit}(v)$ represents the fitted parametric PDF.
\end{itemize}

All metrics are evaluated on a common discretized grid over the support to ensure consistency and comparability. The corresponding quantitative results are reported in Table~\ref{tab:error_metrics}, which also includes the moment-matching parameter estimates derived in this work.

\begin{table*}[!t]
\caption{Quantitative comparison of fitting accuracy for representative Gamma-type approximations of Poisson--Voronoi typical-cell volume distributions under a unified evaluation framework.}
\centering
\scriptsize
\begin{tabular}{cccccc}
\toprule
Dimension & Model & Parameters / Reference & MSE & MAE & $R^2$ \\
\midrule

\multirow{8}{*}{2D}
& One-parameter Gamma
& $d=2$~\cite{ferenc2007size}
& 0.00001946 & 0.00251003 & 0.99975298 \\
\cmidrule(lr){2-6}

& \multirow{3}{*}{Two-parameter Gamma}
& $a=b=3.61$~\cite{weaire1986distribution}
& 0.00002581 & 0.00283045 & 0.99967783 \\
&
& $a=3.61$, $b=3.57$~\cite{dicenzo1989monte}
& 0.00003174 & 0.00304038 & 0.99959723 \\
&
& $a=b=3.5692$ (this work)
& \textbf{0.00001684} & \textbf{0.00234632} & \textbf{0.99978856} \\
\cmidrule(lr){2-6}

& \multirow{3}{*}{Three-parameter Gamma}
& $a=3.3095$, $b=3.0328$, $c=1.0787$~\cite{hinde1980monte}
& 0.00000411 & 0.00117812 & 0.99994769 \\
&
& $a=3.315$, $b=3.04011$, $c=1.078$~\cite{tanemura2003statistical}
& 0.00000380 & 0.00113636 & 0.99995170 \\
&
& $a=3.31605$, $b=3.03689$, $c=1.07871$ (this work)
& \textbf{0.00000350} & \textbf{0.00110210} & \textbf{0.99995544} \\

\midrule

\multirow{6}{*}{3D}	
& One-parameter Gamma
& $d=3$~\cite{ferenc2007size}
& 0.00027711 & 0.00920555 & 0.99713148 \\
\cmidrule(lr){2-6}

& \multirow{2}{*}{Two-parameter Gamma}
& $a=b=6$~\cite{kiang1966random}
& 0.00044188 & 0.01216495 & 0.99589027 \\
&
& $a=b=5.5856$ (this work)
& \textbf{0.00012842} & \textbf{0.00638545} & \textbf{0.99875426} \\
\cmidrule(lr){2-6}

& \multirow{2}{*}{Three-parameter Gamma}
& $a=4.8065$, $b=4.06342$, $c=1.16391$~\cite{tanemura2003statistical}
& 0.00004937 & 0.00465366 & 0.99951284 \\
&
& $a=4.828$, $b=4.069$, $c=1.160$ (this work)
& \textbf{0.00003801} & \textbf{0.00404746} & \textbf{0.99962220} \\

\bottomrule
\end{tabular}
\label{tab:error_metrics}
\end{table*}

Several technically meaningful trends can be observed from Table~\ref{tab:error_metrics}.

\subsubsection{Model Complexity vs. Tail Fidelity}

Across 2D and 3D, the GG models consistently outperform the TG models in fitting accuracy. This improvement is not merely numerical: the additional exponent parameter $c$ in \eqref{eq:three_para} provides the flexibility to match the peak region while simultaneously capturing the empirical tail behavior. In contrast, under unit-mean normalization, the TG family effectively retains only one free parameter, limiting its ability to fit both the central mass and the large-volume tail.

\subsubsection{Increased Fitting Difficulty in 3D}
All models exhibit noticeably larger MSE/MAE in 3D than in 2D, reflecting the fact that the 3D Poisson--Voronoi volume distribution is more skewed and more sensitive to tail mismatch. In particular, the one-parameter Ferenc--N\'eda baseline becomes substantially less accurate in 3D, indicating that dimension-only parameterization cannot adequately capture higher-order distributional features (e.g., skewness and tail decay) as the dimension $d$ increases.

\subsubsection{Dimension-Only Modeling: Contrast Between PVTs and PDTs}
The degradation of the Ferenc--N\'eda approximation in 3D also highlights a fundamental gap between the two dual tessellations. For PVTs, purely dimension-driven closed-form surrogates remain empirical and can lose accuracy as $d$ increases. In sharp contrast, for PDTs the typical simplex-volume PDF admits the exact, dimension-explicit representation in \eqref{eq:delaunay_pdf} for arbitrary $d$, governed solely by $d$ and the PPP intensity. This comparison points to an important open direction: constructing dimension-explicit approximation hierarchies for Poisson--Voronoi typical-cell volume distributions, in which the approximating PDF depends explicitly on $d$ and can be systematically refined through additional moment, Mellin-transform, or geometric constraints. Ideally, such approximations should admit provable accuracy control and correct tail behavior, thereby providing a Voronoi-side counterpart to the exact Delaunay formula \eqref{eq:delaunay_pdf}.

\subsubsection{Effectiveness of Moment-Matched Fitting}
The moment-matched parameterizations derived in this work yield the smallest MSE and MAE in all considered cases. This confirms that incorporating benchmark moments beyond the mean (and, for GG, beyond the variance) materially improves the fit. The gains are particularly pronounced in 3D, where matching the third moment helps mitigate the systematic tail underestimation observed in TG models.

Although the absolute differences in MSE and MAE are numerically small, they are not negligible from an engineering standpoint. In wireless networks, rare but extreme realizations of large Voronoi cells disproportionately affect user-load imbalance, interference aggregation, and outage events. Consequently, even modest mismatches in the tail probability $\mathbb{P}(V>v)$ can propagate into noticeable deviations in performance metrics such as rate coverage probability, cell-edge reliability, and load-dependent base-station activity factors.

\subsection{Empirical PDFs and Tail Discrepancies}

\emph{Simulation setup:}
The numerical experiments were implemented in MATLAB R2025a using the built-in computational-geometry routines \code{voronoiDiagram}, \code{delaunayTriangulation}, \code{voronoin}, and \code{convhull}. The Optimization Toolbox was used for parameter estimation through \texttt{fminunc}. For each realization, a homogeneous PPP of intensity $\lambda=1$ was generated by first drawing the number of nuclei in the extended window from a Poisson distribution with mean $\lambda|\mathcal{W}_{\mathrm{ext}}|$ and then sampling their locations independently and uniformly within that window. The nominal observation windows were $[0,100]^2$ and $[0,20]^3$ for the two- and three-dimensional simulations, respectively. With a guard ratio of $0.6$, the corresponding extended windows were $[-60,160]^2$ and $[-12,32]^3$. To mitigate boundary effects, only bounded cells whose nuclei lay within the nominal observation window and whose vertices remained strictly inside the corresponding extended window were retained.

Across $4000$ realizations, the resulting samples contained $40{,}002{,}672$ two-dimensional cells and $32{,}005{,}269$ three-dimensional cells. For deterministic regeneration, MATLAB's Mersenne Twister generator (\texttt{twister}) was initialized with seed $2025$. Because $\lambda=1$, the retained cell volumes were already expressed in the intensity-normalized form $\lambda V=V$, whose theoretical mean is one; no additional sample-mean normalization was applied. The empirical PDFs were constructed using $400$ histogram bins in each dimension. Model parameters were estimated by moment matching. For the three-parameter generalized Gamma model, a quasi-Newton implementation of \texttt{fminunc} minimized the sum of squared discrepancies between the theoretical and empirical first three raw moments. The two-parameter Gamma model was fitted using the first two moments under the unit-mean normalization. The complete code, numerical data, parameter settings, and reproduction instructions are provided in the Code Availability section.

Fig.~\ref{fig:voronoi_pdf_both} compares the empirical PDFs of normalized Poisson--Voronoi cell volumes with several fitted Gamma-type models, highlighting differences in both the peak and large-volume tail regimes. The distributions exhibit the characteristic right-skewed shape of Poisson--Voronoi cell volumes: most cells are smaller than the mean, while a relatively small number of large cells produce a long right tail. This skewness is consistent with the moment statistics summarized in Table~\ref{tab:moment_summary}, where the squared coefficients of variation indicate moderate volume dispersion.

\begin{figure}[!t]
	\centering
	\subfloat[\scriptsize Areas in 2D]{%
		\includegraphics[width=0.475\textwidth]{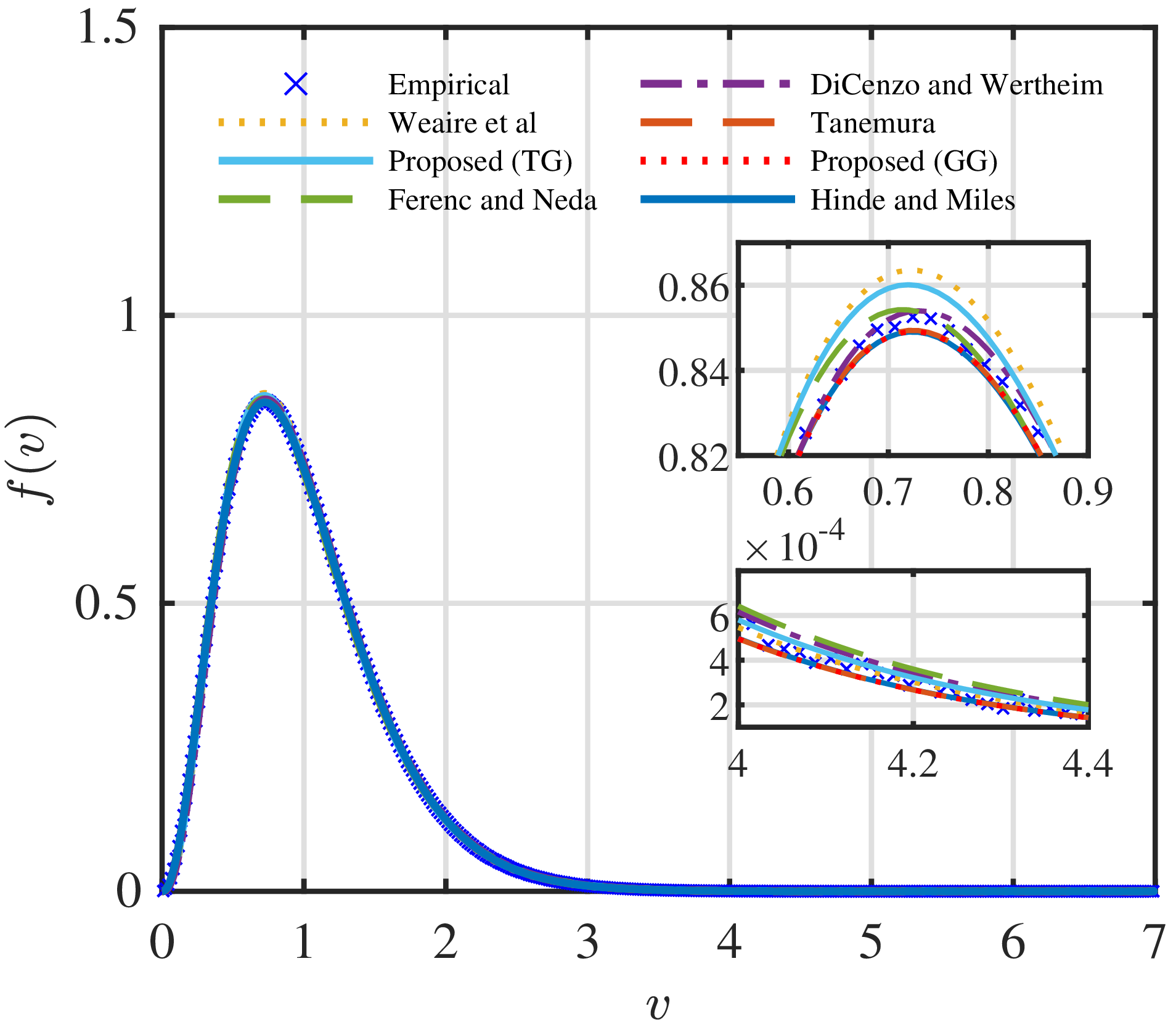}%
		\label{fig:voronoi_pdf_2d}}\\
	\subfloat[\scriptsize Volumes in 3D]{%
		\includegraphics[width=0.475\textwidth]{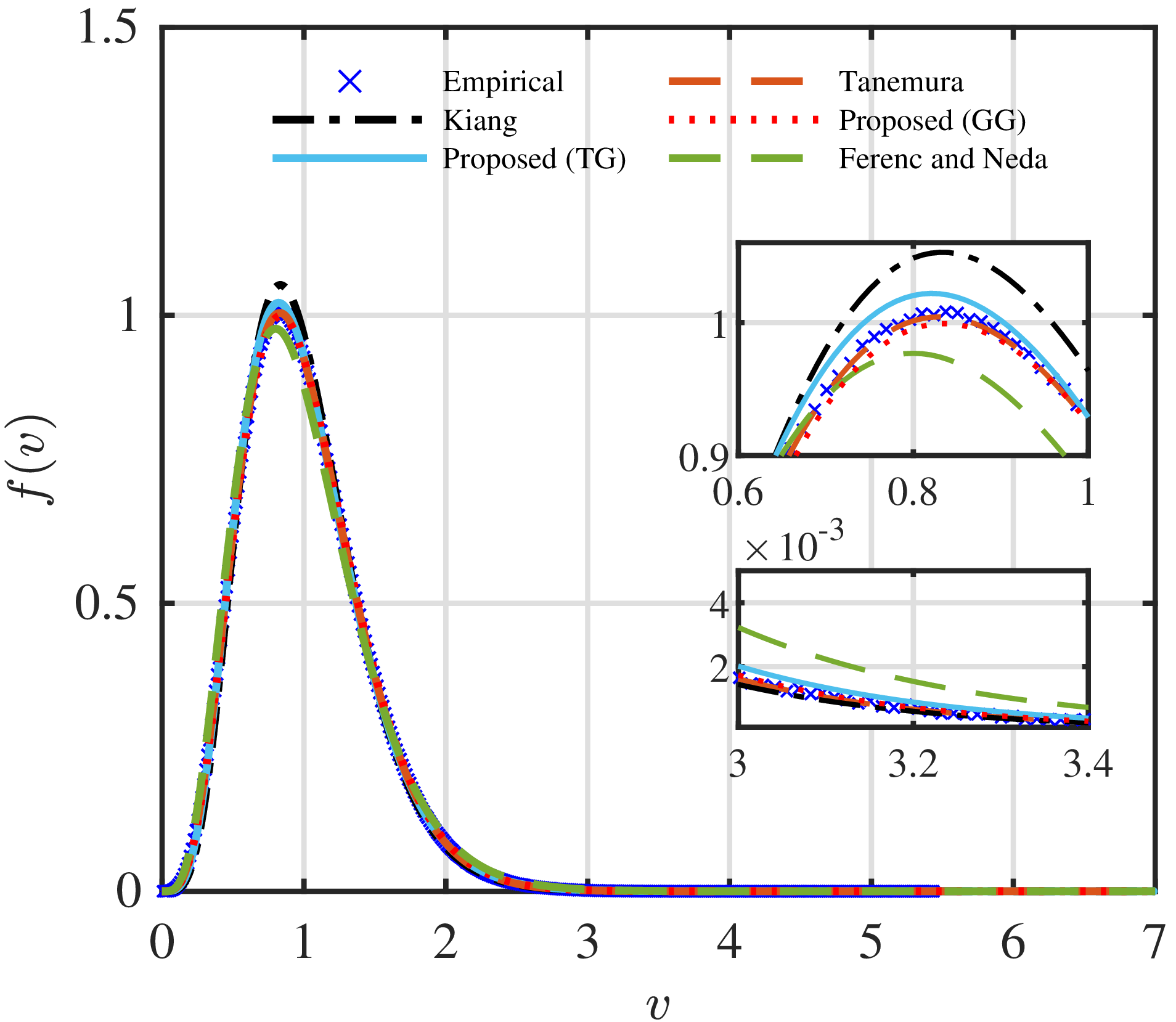}%
		\label{fig:voronoi_pdf_3d}}
	\caption{Empirical and fitted PDFs of normalized Poisson--Voronoi typical-cell volumes: (a) cell areas in 2D and (b) cell volumes in 3D. Markers denote Monte Carlo simulations, while curves correspond to fitted Gamma-type models. The upper inset zooms in on the peak region to compare mode alignment, while the lower inset highlights the large-volume tail, where differences in decay behavior become more apparent.}
	\label{fig:voronoi_pdf_both}
\end{figure}

The peak insets show that most Gamma-type parameterizations reproduce the peak location and central region reasonably well, which helps explain the uniformly high $R^2$ values reported in Table~\ref{tab:error_metrics}. However, the main panels reveal more pronounced discrepancies in the large-volume regime, particularly for $v\gtrsim2\,\mathbb{E}[V]$. The tail insets expose the principal distributional differences. In 2D, the fitted curves exhibit a clear ordering of tail heaviness, indicating that different parameterizations impose noticeably different large-volume decay rates even when their central fits are similar. In 3D, this effect becomes more pronounced: the Ferenc--N\'eda model remains systematically above the empirical markers across the plotted tail range, whereas the moment-matched TG and GG models remain closer to the simulated results. These observations indicate that the tail regime provides a stringent test of model fidelity because small PDF differences can produce materially different probabilities of unusually large Voronoi cells.

From a stochastic-geometry perspective, accurate tail modeling is important because large Voronoi cells represent unusually large service regions and may be associated with heavier user loads. When combined with models of user locations, cell shape, association, propagation, fading, and interference, the tail probability $\mathbb{P}(V>v)$ can therefore influence predictions of load imbalance, rate coverage, and network reliability. Cell volume alone, however, does not determine user-to-serving distances, signal strength, or wireless-system performance. The tail results should accordingly be interpreted as geometric inputs to such analyses rather than as standalone performance predictions.

\subsection{Key Insights for Stochastic-Geometry Modeling}

The empirical comparisons and moment analyses presented above reveal several structural properties of PVTs and PDTs that are directly relevant for stochastic-geometry modeling. The main insights can be summarized as follows.

\subsubsection{Voronoi and Delaunay Typical-Cell Volumes Exhibit Fundamentally Different Dimension Dependence}

For PVTs, the typical-cell volume remains normalized at $\mathbb{E}[V]=1/\lambda$ for all $d$, reflecting the space-partition property that each PPP point ``owns'' an average region of volume $1/\lambda$. In contrast, Poisson--Delaunay simplex volumes decrease rapidly with dimension (Table~\ref{tab:moment_summary}), consistent with the empty-circumsphere constraint and the super-polynomial high-dimensional scaling in \eqref{eq:delaunay_typical_scale_final}. This contrast highlights that Voronoi geometry captures \emph{coverage dominance}, whereas Delaunay geometry captures \emph{local adjacency}, and these two notions behave very differently as the dimension increases.

\subsubsection{For Poisson--Voronoi Modeling, Three-Parameter Models Improve Tail Fidelity}

The exact scale--shape factorization in Section~\ref{subsubsec:scale_shape} provides a structural explanation for both the effectiveness and limitations of Gamma-type approximations. Conditional on the effective-facet count, the cell volume contains an independent Gamma scale component. However, variability in the normalized cell-to-flower ratio and mixing across facet counts prevent the unconditional volume distribution from following a single Gamma law. The additional flexibility of the GG model can therefore be interpreted as an empirical means of accommodating these unresolved sources of geometric variability.

For $d=2,3$, the GG models considered here provide the most reliable approximations of typical-cell volumes, particularly in the large-volume regime relevant to load imbalance and cell-edge reliability. Simpler TG models remain attractive when analytical simplicity is paramount and generally reproduce the peak region well, but their flexibility in capturing tail behavior across dimensions is limited. As reflected in the error hierarchy in Table~\ref{tab:error_metrics} and the tail comparisons in Fig.~\ref{fig:voronoi_pdf_both}, the exponent parameter in the GG model provides improved control of large-volume decay. Consequently, for tail-sensitive metrics such as rate coverage and outage probability, the GG model offers the best simplicity--accuracy tradeoff among the models examined here.

\subsubsection{Poisson--Delaunay Statistics Provide an Exact and Scalable Benchmark for Adjacency-Based Modeling}

Unlike the Voronoi case, Poisson--Delaunay simplex volumes admit exact, dimension-explicit PDFs and CDFs for arbitrary $d$ through Meijer's $G$-function representations; see \eqref{eq:delaunay_pdf}. These volume statistics can therefore serve as exact distributional inputs to stochastic-geometry-based performance models rather than relying on empirical fitting. Geometrically, the rapid contraction of typical simplex volumes characterizes the increasingly localized scale of Delaunay adjacency in the mathematical high-dimensional regime. Its implications for practical cooperation mechanisms must nevertheless be assessed separately in finite-dimensional network models that also incorporate propagation, association, and resource-allocation effects. Moreover, the availability of exact moments and closed forms provides a rigorous benchmark for validating Monte Carlo simulations and developing tractable approximations for more general non-Poisson or dynamic settings.

In summary, Poisson--Voronoi volumes admit an exact scale--shape representation that reveals their conditional Gamma structure, although practical evaluation of their unconditional distribution still relies largely on moment-based and empirical approximations. Poisson--Delaunay simplex volumes, by contrast, admit exact dimension-explicit PDFs, CDFs, and moments. This analytical asymmetry provides a geometric foundation for stochastic-geometry models in wireless communications: Voronoi cells represent coverage and load regions, whereas Delaunay simplices capture adjacency relations underlying cooperation and clustering. The next section discusses how these complementary structures are used in wireless-network analysis and design.

%
\section{Representative Wireless-Network Applications of Voronoi and Delaunay Tessellations}
\label{sec:application}

PVTs and PDTs provide geometric abstractions for stochastic-geometry-based modeling of wireless networks. In particular, PVTs naturally describe coverage regions and load distribution under nearest-BS association, whereas PDTs encode adjacency relations used in modeling cooperative transmission and clustering. The analytical and numerical results presented in Section~\ref{sec:results} show that the typical-cell and typical-simplex volume statistics of PVTs and PDTs exhibit different distributional behaviors. These statistics are relevant to models of user load, rate coverage, BS activity, cooperative association, and network resilience.

Typical-cell and typical-simplex volumes should be interpreted as geometric inputs to wireless-network models rather than as standalone performance measures. Handover behavior, connectivity, interference, rate coverage, and spectral efficiency depend jointly on additional factors, including link distances, cell shape, adjacency relations, propagation loss and fading, association policies, resource allocation, and node mobility. Consequently, volume statistics can support the modeling of these quantities but cannot, by themselves, establish system-performance improvements.

In emerging 5G/6G ultra-dense networks, Poisson--Voronoi cell-volume distributions provide inputs for modeling user load, void-BS probabilities, and load-dependent rate coverage~\cite{9079449,8515110,6042301,7389350}. In UAV-assisted and integrated three-dimensional networks, Voronoi volumes contribute to the characterization of spatial coverage and load allocation~\cite{8437232,11361371}, whereas Delaunay simplex volumes, when combined with adjacency and mobility models, can characterize the geometric scale of cooperation clusters and support analyses of mobility stability and local connectivity~\cite{8358978,9151343,9893878,10720695,11311373,marchand2024cox,10381632,11346534}. These complementary roles make cell- and simplex-volume statistics relevant to terrestrial, aerial, and vertically integrated network architectures.

In the following subsections, we review representative wireless-network applications of both tessellation models and highlight their complementary roles in system-level stochastic-geometry analysis.

\subsection{Voronoi-Based Modeling of Coverage, Load, and Resource Sharing}

The PDF of the typical Poisson--Voronoi cell volume (area in two-dimensional networks) is a fundamental geometric ingredient in stochastic-geometry-based cellular network analysis. Its applications span several interrelated aspects of wireless-network modeling.

\subsubsection{User Load and Users-Per-Cell Distribution}

Under nearest-BS association and a homogeneous PPP user model, the number of users served by a BS is governed by the volume $V$ of its Voronoi cell. Conditional on $V$, the number of users follows a Poisson distribution with mean proportional to $V$, leading to a mixed-Poisson users-per-cell distribution. This framework has been widely used to analyze user-load distributions and traffic heterogeneity in large-scale cellular networks~\cite{6042301,9079449,8515110,6576422}. Such results provide tractable inputs for load estimation, access control, and mobility management.

\subsubsection{Rate Coverage and Load Coupling}

Because radio resources are shared among all active users in the same Voronoi region, the user rate is inherently coupled to cell load. Consequently, the Voronoi volume distribution directly influences the achievable rate distribution across the network. By combining Voronoi volume statistics with signal-to-interference-plus-noise ratio (SINR) distributions, tractable approximations of rate coverage probability and rate distributions can be obtained \cite{6042301}. This methodology has been widely applied to heterogeneous cellular networks~\cite{6497002,6684550} and millimeter-wave systems~\cite{7593259,9625443}, where blockages and directional beamforming reshape effective serving regions and thereby alter load distributions.

\subsubsection{Void Cells, BS Activity, and Energy-Efficient Operation}

The probability that a BS is \emph{void} (i.e., has no associated users) depends directly on the distribution of Voronoi cell volumes, leading to the notion of a BS activity factor. Incorporating void-cell effects yields more accurate interference characterization, particularly in lightly loaded or ultra-dense regimes. As shown in~\cite{7389350,8358978,8515110}, neglecting void probabilities can substantially overestimate interference and distort coverage predictions. Voronoi volume statistics also enable load-aware BS sleep strategies, whereby a BS is deactivated when its associated user count falls below a threshold. The resulting trade-off between reduced interference and coverage degradation has been studied in both homogeneous and heterogeneous deployments~\cite{9676414} and in UAV-assisted cellular systems~\cite{8437232}. Consequently, volume-based load modeling provides a tractable framework for analyzing energy--coverage trade-offs.

\subsubsection{Finite-Region Effects and Security Analysis}

In bounded regions, Voronoi cells are truncated by network boundaries, resulting in location-dependent volume distributions. This boundary effect has been exploited in secure connectivity analysis, where users near the edges of finite regions exhibit distinct geometric statistics and, consequently, different secrecy and connection probabilities~\cite{Koufos2019,5995290}. More broadly, stochastic-geometry frameworks have been widely used for physical-layer security analysis in large-scale PPP-based wireless networks~\cite{6512533,9756566}. Consequently, Voronoi volume statistics play an important role not only in load and coverage modeling but also in evaluating secrecy performance and finite-region network behavior.

Overall, the typical-cell volume distribution of PVTs remains a cornerstone of stochastic-geometry-based wireless modeling, supporting tractable analysis of load distribution, rate coverage, interference mitigation, energy efficiency, and secure connectivity.

These applications illustrate how Voronoi cell volumes provide a fundamental geometric descriptor of coverage regions and traffic load in cellular networks. While Voronoi tessellations primarily support coverage-oriented modeling under nearest-BS association, many emerging wireless architectures increasingly rely on cooperative transmission and topology-aware networking. These mechanisms depend on adjacency relations among neighboring nodes, which are naturally captured by PDTs.

\subsection{Delaunay-Based Modeling of Cooperation and Network Topology}

The typical Poisson--Delaunay simplex volume provides a natural geometric abstraction for analyzing multi-point association, CoMP transmission, mobility-aware networking, and topology robustness. Unlike the Voronoi case, the simplex-volume distribution admits an exact, dimension-explicit characterization via Meijer's-$G$ representations for arbitrary $d$; see~\eqref{eq:delaunay_pdf}. This analytical tractability makes Delaunay-based modeling particularly attractive for structured cooperation and adjacency-driven wireless design.

\subsubsection{Cooperative Set Formation and Multi-BS Association}

As the geometric dual of PVTs, PDTs naturally encode BS adjacency: each simplex corresponds to a minimal cluster of $d+1$ neighboring BSs satisfying the empty-circumsphere condition. The simplex volume provides a geometric measure of the spatial scale of such adjacency clusters and reflects the local spatial density of cooperating BSs. Leveraging this structure, Poisson--Delaunay-based CoMP transmission frameworks have been proposed for planar cellular systems~\cite{8358978,8976426} and extended to air-to-air networks~\cite{9151343}, integrated ground-to-air~\cite{9893878} networks, and vertical heterogeneous networks \cite{11311373}. In these settings, simplex geometry provides a principled mechanism for systematically defining cooperative sets that are both spatially consistent and analytically tractable.

\subsubsection{Coverage and Rate Analysis under Cooperative Clustering}

The gains achievable through CoMP transmission depend critically on the spatial configuration of cooperating BSs. Delaunay simplices provide a natural representation of local cooperation geometry. By combining simplex-volume statistics with inter-point distance distributions and SINR analysis, tractable expressions for coverage probability and rate performance under adjacency-based cooperation can be derived. This approach has been applied in terrestrial networks~\cite{8976426}, low-altitude aerial systems~\cite{9893878,9151343}, and coordinated wireless power transfer architectures~\cite{9775573}, demonstrating that geometry-consistent clustering can significantly improve spectral efficiency.

\subsubsection{Mobility, Handoff, and Interference Coordination}

A key advantage of Delaunay-based association is its geometric consistency: cooperation sets are determined by stable adjacency relations rather than instantaneous received power. This reduces frequent re-association under mobility and improves network stability. By exploiting Delaunay geometry and simplex-volume statistics, analytical characterizations of handoff behavior and interference coordination can be derived, thereby improving robustness in UAV-enabled networks~\cite{9893878,10720695} and air-to-air systems~\cite{9151343}.

\subsubsection{Topology Control, Connectivity, and Network Resilience}

Because Delaunay simplices explicitly encode local neighbor relations among nodes, they provide a natural geometric foundation for topology-aware coordination in wireless networks. In particular, simplex-based adjacency structures enable principled formation of cooperative clusters and support coordinated scheduling, power control, and distributed resource management. Compared with Voronoi abstractions that primarily describe coverage dominance, Delaunay simplices more directly capture the geometric footprint of cooperative groups, thereby improving spectral efficiency and reducing transmit power in coordination-enabled systems. Geometry-aware clustering based on Delaunay adjacency has therefore been explored in several emerging wireless architectures, including CoMP-enabled aerial networks~\cite{9151343}, vertically integrated heterogeneous systems~\cite{11311373}, and reconfigurable intelligent surfaces (RIS)-assisted three-dimensional cellular deployments~\cite{11361371}.

Beyond coordination and clustering, Delaunay geometry also provides a principled framework for studying connectivity and resilience in spatial wireless networks. The Poisson--Delaunay graph forms a proximity-based backbone whose connectivity and percolation properties have been rigorously analyzed in stochastic spatial models~\cite{marchand2024cox}. In engineering applications, similar proximity structures naturally describe distributed neighbor relations, such as those arising in UAV swarm formation and tracking~\cite{10381632} and UAV-assisted IoT coverage recovery under dynamic deployments~\cite{11346534}. While global connectivity is primarily governed by the edge structure of the Delaunay graph, simplex-volume statistics provide complementary geometric insight. Large simplex volumes indicate locally sparse node configurations that may correspond to weakly connected regions of the Delaunay graph and potential geometric bottlenecks, whereas smaller volumes correspond to dense local neighborhoods. Consequently, simplex-volume distributions can serve as useful indicators for assessing robustness, fragmentation risk, and topology stability in Delaunay-based wireless networks.

In summary, the typical-cell volume statistics of PVTs and PDTs provide complementary geometric abstractions for stochastic-geometry-based wireless-network modeling. Voronoi volumes primarily characterize coverage regions, load distribution, and interference activity under nearest-BS association, whereas Delaunay simplices encode adjacency relations that enable principled modeling of cooperative transmission, clustering, mobility management, and topology resilience. The availability of exact, dimension-explicit simplex-volume distributions further enhances the analytical appeal of PDT-based abstractions for adjacency-driven wireless design. At the same time, the absence of comparable dimension-explicit characterizations for Voronoi volumes highlights an important theoretical gap that motivates the research directions discussed in the next section.

\section{Future Research Opportunities}
\label{sec:future_research}

The wireless-network applications reviewed in Section~\ref{sec:application} highlight the central role of typical-cell volume statistics in stochastic-geometry-based system modeling. Despite substantial progress in understanding these geometric quantities, several fundamental challenges remain. In particular, the preceding sections showed that Poisson--Delaunay simplex volumes admit directly evaluable, dimension-explicit characterizations, whereas the recently developed scale--shape factorization provides an exact but still implicit representation of Poisson--Voronoi volumes. The remaining challenge is to convert this structural representation into tractable unconditional distributions or systematically controlled approximations. Bridging this gap, extending the theory beyond homogeneous PPPs, and adapting tessellation models to emerging wireless architectures therefore represent important directions for future research.

The following research directions can be broadly organized into three complementary themes.
First, \emph{analytical directions} focus on deepening the mathematical understanding of Voronoi and Delaunay volume distributions, including closed-form characterizations, tail asymptotics, and high-dimensional limit behavior.
Second, \emph{modeling directions} aim to extend tessellation-based frameworks beyond idealized Poisson settings by incorporating spatial correlations, heterogeneous intensities, and temporal dynamics that more accurately reflect practical wireless deployments.
Third, \emph{computational directions} emphasize scalable evaluation and approximation techniques, including special-function asymptotics and hybrid analytic--data-driven surrogates that preserve stochastic-geometry invariances.
Together, these themes define a coherent research agenda to advance the theory of random spatial tessellations and to enable more accurate and computationally tractable stochastic-geometry models for next-generation wireless networks.

\subsection{Closed-Form Distributions for Poisson--Voronoi Cells}

A long-standing open problem in stochastic geometry is the absence of a tractable unconditional closed-form PDF for the typical Poisson--Voronoi cell volume, even in the fundamental cases \(d=2\) and \(d=3\). Exact planar integral formulae are available in \(d=2\)~\cite{Calka2003}, and the recent scale--shape factorization provides an exact structural and mixture representation in arbitrary dimensions~\cite{xia2026scaleshape}. Conditional on the effective-facet count, it separates an independent Gamma scale from a bounded normalized shape factor. Nevertheless, the conditional shape laws and the unbounded facet-count mixture remain implicit, preventing direct evaluation of the unconditional distribution.

Several technically promising directions follow from this representation:
\begin{itemize}
	\item \emph{Normalized shape laws:} Characterizing or accurately approximating the conditional distributions of the cell-to-flower volume ratio for fixed effective-facet counts;
	\item \emph{Controlled facet-count truncation:} Developing finite-mixture approximations with explicit error bounds for the neglected higher-facet contributions;
	\item \emph{Tail asymptotics:} Establishing unconditional lower- and upper-tail behavior by combining configuration-space analysis with uniform control across facet counts;
	\item \emph{Moment and transform inversion:} Using the exact conditional identities together with stable inversion techniques to construct systematically improvable distributional approximations.
\end{itemize}

Even partial analytical advances, particularly improved tail approximations, would significantly enhance load-aware wireless-network models.

\subsection{Dimension-Explicit Distributions and Systematic Approximations}

Sections~\ref{sec:PDFs_of_volumes} and~\ref{sec:results} highlight a sharp analytical contrast. Poisson--Delaunay simplex volumes admit directly evaluable, dimension-explicit PDFs and CDFs for arbitrary $d$ through Meijer's $G$-function representations; see~\eqref{eq:delaunay_pdf}. For Poisson--Voronoi volumes, the exact scale--shape factorization is valid in arbitrary dimensions but remains implicit through the conditional shape laws and effective-facet-count mixture. Practical evaluation therefore continues to rely primarily on simulation benchmarks and Gamma-type approximations. Moreover, dimension-only surrogates such as the Ferenc--N\'eda model exhibit noticeable degradation in $d=3$, particularly in the tail regime.

A key open problem is therefore to develop either exact \emph{dimension-explicit} distributions, if feasible, or \emph{systematically improvable} approximations for the typical Poisson--Voronoi volume that remain accurate across dimensions. Ideally, such approximations should satisfy:
\begin{itemize}
	\item \emph{Uniform validity across $d$:} reliable accuracy from low to moderate dimensions;
	\item \emph{Exact PPP scaling:} preservation of the scaling law in \eqref{eq:voronoi_scaling_dist};
	\item \emph{Moment consistency:} agreement with benchmark moments by construction;
	\item \emph{Tail fidelity:} accurate asymptotic decay of $\mathbb{P}(V>v)$ and $\mathbb{P}(V<v)$.
\end{itemize}

Promising approaches include finite-mixture approximations derived from the scale--shape factorization, parametric or data-driven models for the conditional shape laws, Mellin-transform-based approximations, maximum-entropy constructions under moment constraints, and hybrid analytic--data-driven surrogates that enforce stochastic-geometry invariances while providing quantifiable error control.

\subsection{Finite-Dimensional Accuracy and Practical Scaling Laws}

The high-dimensional limit theory summarized in Corollary~\ref{corollary:Zd_limit_theory} establishes strong concentration of Poisson--Delaunay simplex volumes under logarithmic normalization. However, most practical wireless-network applications operate in low or moderate dimensions ($d=2$ or $d=3$), where asymptotic limits alone do not fully capture finite-dimensional behavior.

An important direction, therefore, is to bridge the gap between asymptotic theory and practically relevant dimensions. Key research topics include:
\begin{itemize}
\item \emph{Finite-$d$ error bounds:} Quantifying convergence rates of high-dimensional approximations;
\item \emph{Uniform approximations:} Developing analytic models that remain accurate across a range of dimensions;
\item \emph{Finite-dimensional geometric benchmarks:} Translating log-volume concentration into controlled approximations for volume-dependent inputs to coverage, load, and cooperative-cluster models.
\end{itemize}

Such results would substantially enhance the practical applicability of high-dimensional stochastic-geometry theory.

Beyond these foundational analytical questions, an equally important challenge is extending tessellation-based models to the more realistic spatial environments encountered in practical wireless systems.

\subsection{Extension Beyond Homogeneous PPPs}

Most existing results rely critically on the independence and stationarity properties of homogeneous PPPs. Real wireless deployments, however, often exhibit spatial repulsion (planned BS layouts), clustering (user hotspots), or spatial inhomogeneity (urban--rural gradients). Extending typical-cell volume theory beyond Poisson assumptions, therefore, represents an important and technically demanding research direction.

The principal obstacle is the loss of Palm-calculus tractability: conditional independence structures that simplify PPP analysis no longer hold. Promising research directions include:
\begin{itemize}
	\item \emph{Clustered models:} Cox, Thomas, and Neyman--Scott processes capturing hotspot behavior~\cite{Haenggi2012,Chiu2013};
	\item \emph{Repulsive models:} determinantal or Mat\'ern hard-core processes reflecting planned deployments~\cite{Haenggi2012};
	\item \emph{Inhomogeneous models:} spatially varying intensities relevant to heterogeneous networks~\cite{Chiu2013};
	\item \emph{Deterministic and aperiodic models:} Delone sets, lattices, and quasicrystalline point patterns that capture geometrically regular yet non-Poisson spatial organization~\cite{Lagarias1999,BaakeGrimm2013}.
\end{itemize}

Developing approximation frameworks that preserve analytical tractability while capturing spatial correlations remains a central open challenge.

\subsection{Dynamic and Time-Varying Tessellations}

In mobile scenarios such as vehicular networks, UAV-assisted systems, and mobile IoT deployments, the underlying point process evolves over time, inducing time-varying Voronoi and Delaunay tessellations. The typical-cell volume therefore becomes a stochastic process $\{V(t)\}$ rather than a static random variable.

Understanding the temporal behavior of such tessellations remains largely unexplored. Important questions include:
\begin{itemize}
	\item \emph{Temporal correlation:} Characterizing $\mathrm{Corr}(V(t),V(t+\tau))$ under realistic mobility models;
	\item \emph{Handoff dynamics:} Relating Voronoi-volume fluctuations to boundary crossings and handoff rates;
	\item \emph{Simplex persistence:} Quantifying the temporal stability of Delaunay simplices and its impact on CoMP transmission.
\end{itemize}

Progress in this direction would significantly improve the realism of tessellation-based wireless models.

\subsection{Asymptotic Analysis of Meijer's-$G$ Representations}

The exact PDF of the Poisson--Delaunay simplex volume admits a unified Meijer's-$G$ representation; see~\eqref{eq:delaunay_pdf}. While this expression provides a complete analytic description for all dimensions, its asymptotic behavior for large $d$ has not yet been systematically studied from the perspective of special-function analysis.

Future research could explore:
\begin{itemize}
\item \emph{Saddle-point and pole analysis:} Identifying dominant poles or saddle points governing the large-$d$ asymptotics of the Mellin transform;
\item \emph{Asymptotic evaluation of Meijer's-$G$ functions:} Developing systematic approximations for moderate and high dimensions;
\item \emph{Connections to geometric probability:} Linking special-function asymptotics with concentration phenomena in high-dimensional random tessellations.
\end{itemize}

Such analyses may reveal more clearly how scaling laws in high dimensions emerge directly from the exact representation in terms of special functions.

Finally, the increasing scale and complexity of modern wireless systems motivate the development of efficient computational and data-driven tools that complement traditional analytical approaches.

\subsection{Hybrid Analytic--Data-Driven Modeling}

The implicit nature of the exact Poisson--Voronoi scale--shape representation and the numerical complexity of Meijer's $G$-functions motivate hybrid analytic--data-driven approaches. Unlike purely black-box machine-learning methods, geometry-aware surrogates can preserve known stochastic-geometry structure, including PPP intensity scaling, moment identities, positivity, support bounds, and conditional mixture structure.

Promising directions include:
\begin{itemize}
\item Learning conditional shape-ratio distributions subject to the support bound in \eqref{eq:PV_shape_support};
\item Constructing finite facet-count mixtures that preserve exact moments and intensity scaling;
\item Developing normalizing-flow models that respect positivity and geometric constraints;
\item Using neural-network surrogates to accelerate the evaluation of Meijer's $G$-function PDFs in moderate and high dimensions.
\end{itemize}

Balancing computational efficiency with theoretical consistency remains the key challenge, ensuring that learned models respect geometric probability structure rather than merely fitting simulation data.

In summary, advancing the theory of Poisson--Voronoi and Poisson--Delaunay typical-cell volume distributions requires progress along several complementary directions. On the theoretical side, the exact Poisson--Voronoi scale--shape factorization provides a new foundation for characterizing conditional shape laws, controlling facet-count mixtures, and deriving unconditional tail asymptotics and dimension-explicit approximations. On the modeling side, extending tessellation frameworks beyond homogeneous PPPs and incorporating temporal dynamics will be essential for representing the spatial and mobility patterns of emerging wireless systems. From a computational perspective, hybrid analytic--data-driven approaches offer a promising means of translating rigorous structural results into scalable evaluation and approximation tools. Progress in these directions will deepen our understanding of random spatial tessellations and enable more accurate performance analysis for next-generation wireless networks operating in complex, dynamic environments.

Based on the preceding discussion, several representative open problems can be summarized as follows.
\begin{remark}[Key Open Problems]
Despite substantial progress in the stochastic-geometry analysis of Voronoi and Delaunay tessellations, several fundamental problems remain open:
\begin{enumerate}

\item \emph{Closed-form Poisson--Voronoi volume distributions:}
Deriving explicit PDFs, or provably accurate analytical approximations, for the typical Poisson--Voronoi cell volume, even in the fundamental cases $d=2$ and $d=3$, remains a long-standing challenge. Recent work~\cite{xia2026scaleshape} has established an exact scale--shape factorization: conditional on the effective-facet count, the Voronoi flower volume follows a Gamma distribution and is independent of the normalized shape, yielding exact mixture representations, transform identities, and moment identities. Nevertheless, obtaining a tractable unconditional cell-volume distribution still requires mixing over the unbounded range of the effective-facet count and characterizing the associated normalized shape factors. Developing controlled truncation schemes, explicit error bounds, and unconditional lower-tail asymptotics therefore remains an important open problem.

\item \emph{Dimension-explicit approximation frameworks:}
Systematic approximation schemes are needed for Voronoi volumes that remain accurate across spatial dimensions while preserving exact PPP scaling laws. Existing dimension-dependent approximations may lose accuracy in moderate dimensions or in the distribution tails. A key challenge is therefore to construct systematically improvable models that enforce scaling, moment consistency, and tail fidelity while providing quantifiable error control.

\item \emph{Non-Poisson spatial models:}
Typical-cell volume theory must be extended to clustered, repulsive, and inhomogeneous spatial point processes, as well as to geometrically regular configurations such as Delone sets, lattices, and aperiodic point patterns. Unlike homogeneous PPPs, these models may introduce correlation length scales that cannot be removed by intensity normalization alone. Consequently, approximation methods must account for model-specific interaction structures and be validated across different clustering and repulsion regimes.

\item \emph{Dynamic tessellation theory:}
The temporal evolution of Voronoi and Delaunay tessellations under node mobility remains insufficiently understood, particularly regarding handoff dynamics, cooperative-transmission stability, and network reliability. A central difficulty is that continuous node motion can induce discrete changes in cell boundaries and Delaunay adjacency. Future work should characterize temporal volume correlations, boundary-crossing events, and simplex persistence under realistic mobility models.

\item \emph{Scalable evaluation and data-driven methods:}
Scalable techniques are needed for evaluating complex special-function representations, such as Poisson--Delaunay simplex-volume PDFs expressed through Meijer's $G$-function. For data-driven models, an important unresolved issue is out-of-distribution generalization across spatial dimensions, intensities, and point-process families. Reliable approaches should preserve known geometric constraints, provide uncertainty estimates, and maintain accuracy in the distribution tails, where many wireless-performance metrics are most sensitive.

\end{enumerate}
\end{remark}

The problems highlighted above illustrate that substantial opportunities remain at the intersection of stochastic geometry, high-dimensional probability, and spatial-network modeling. Progress in these directions would not only deepen the mathematical understanding of random spatial tessellations but also provide more accurate and scalable analytical tools for applications ranging from wireless network design to spatial data analysis, computational geometry, and other emerging domains involving large-scale spatial systems.

%

\section{Conclusion}
\label{sec:conclusion}

This paper has presented a structured survey of typical-cell volume distributions in Poisson--Voronoi and Poisson--Delaunay tessellations, emphasizing their analytical foundations, high-dimensional behavior, and relevance to stochastic-geometry-based wireless-network modeling. A central theme is the pronounced analytical asymmetry between these dual tessellation families. Poisson--Voronoi cell volumes admit exact planar integral formulae and, more recently, an exact general-dimensional scale--shape representation with conditional Gamma structure. Nevertheless, the normalized shape laws and unbounded facet-count mixture remain implicit, so practical evaluation of the unconditional distribution still relies largely on simulations, moment identities, and empirical approximations. By contrast, Poisson--Delaunay simplex volumes admit directly evaluable, dimension-explicit PDFs, CDFs, and closed-form moment formulas. More broadly, typical-cell volume distributions provide geometric inputs that connect microscopic spatial randomness with system-level wireless-network models. This survey is intended to serve as a useful reference and stimulate further research on stochastic geometry, random spatial structures, and their wireless-network applications.

%
\appendix

\section{Proof of Corollary~1}
\label{app:1}

Let
\begin{equation}
Y_d \triangleq \log V_d,
\qquad
\widetilde{Y}_d
\triangleq
\frac{Y_d-\mathbb{E}[Y_d]}
{\sqrt{\operatorname{Var}(Y_d)}}.
\end{equation}
The typical Poisson--Delaunay simplex corresponds to the weight parameter
$\mu=-1$ in~\cite{AHL_2021}. For this choice, the expectation and variance
asymptotics are given by
\eqref{eq:delaunay_typical_scale_ahl2021} and
\eqref{eq:delaunay_typical_scale_ahl2021var}, respectively, while
\begin{equation}
\widetilde{Y}_d
\xrightarrow{\mathcal{D}}
\mathcal{N}(0,1)
\end{equation}
follows from~\cite[Th.~1.1(ii)]{AHL_2021}; see also
\cite[Th.~4.4(iv)]{AHL_2021}.

From the definition of $Z_d$,
\[
Z_d
=
\frac{-\log\lambda-Y_d}{d\log d},
\]
and hence
\begin{equation}
\label{eq:Zd_relations}
\mathbb{E}[Z_d]
=
\frac{-\log\lambda-\mathbb{E}[Y_d]}{d\log d},
\quad
\operatorname{Var}(Z_d)
=
\frac{\operatorname{Var}(Y_d)}{d^2\log^2 d}.
\end{equation}

\noindent
(i) Law of Large Numbers:

Equations~\eqref{eq:delaunay_typical_scale_ahl2021}--%
\eqref{eq:delaunay_typical_scale_ahl2021var} and
\eqref{eq:Zd_relations} yield
\begin{align}
\mathbb{E}[Z_d]
&=
\frac12+\mathcal{O}\!\left(\frac{1}{\log d}\right),\\
\operatorname{Var}(Z_d)
&=
\mathcal{O}\!\left(\frac{1}{d^2\log d}\right).
\end{align}
Therefore, $\mathbb{E}[Z_d]\to\tfrac12$ and
$\operatorname{Var}(Z_d)\to0$. Chebyshev's inequality then gives
\[
Z_d\xrightarrow{P}\tfrac12.
\]

\noindent
(ii) Central Limit Theorem:

The definitions of $Y_d$ and $Z_d$ imply the exact identity
\begin{equation}
\label{eq:Zd_standardized}
\frac{Z_d-\mathbb{E}[Z_d]}
{\sqrt{\operatorname{Var}(Z_d)}}
=
-\widetilde{Y}_d.
\end{equation}
The standardized CLT for $\widetilde{Y}_d$, together with the symmetry of
the standard normal distribution, therefore yields
\[
\frac{Z_d-\mathbb{E}[Z_d]}
{\sqrt{\operatorname{Var}(Z_d)}}
\xrightarrow{\mathcal{D}}
\mathcal{N}(0,1).
\]
Furthermore, \eqref{eq:delaunay_typical_scale_ahl2021var} and
\eqref{eq:Zd_relations} give
\[
\sqrt{\operatorname{Var}(Z_d)}
=
\frac{1+o(1)}{d\sqrt{2\log d}}.
\]
Applying Slutsky's theorem \cite[Ch.~3]{Billingsley1999} then gives
\[
\sqrt{2\log d}\,d
\bigl(Z_d-\mathbb{E}[Z_d]\bigr)
\xrightarrow{\mathcal{D}}
\mathcal{N}(0,1).
\]

\noindent
(iii) Moderate Deviation Principle:

For the typical simplex, corresponding to $\mu=-1$,
\cite[Th.~4.4(iii)]{AHL_2021} states that if $a_d\to\infty$ and
$a_d\varepsilon_d^{-1}\to0$, then
$a_d^{-1}\widetilde{Y}_d$ satisfies a moderate deviation principle
with speed $a_d^2$ and rate function $I(x)=x^2/2$. In the fixed-$\mu$
regime, $\varepsilon_d=\sqrt{\log d}$, so the required condition is
equivalent to
\[
a_d=o(\sqrt{\log d}).
\]

Using \eqref{eq:Zd_standardized}, we have
\[
\frac{Z_d-\mathbb{E}[Z_d]}
{a_d\sqrt{\operatorname{Var}(Z_d)}}
=
-\frac{\widetilde{Y}_d}{a_d}.
\]
Since reflection preserves the quadratic rate function $I(x)=x^2/2$, the sequence on the left satisfies the same moderate deviation principle. This completes the proof.

\section*{Acknowledgment}

OpenAI Codex (GPT-5.6 Sol) was used to assist in identifying potentially relevant literature and in the visual preparation of Figs.~1 and 2. It was not used to perform the mathematical analyses or determine the scientific conclusions of this work. All cited literature was independently read, selected, and interpreted by the authors, and the scientific content of Figs.~1 and 2 was independently determined and verified by the authors. The authors take full responsibility for the content of the manuscript.

\section*{Code Availability}
The simulation and plotting code, numerical data, parameter settings, and instructions required to reproduce the numerical results and figures are publicly available at
\href{https://github.com/shitian-0715/Voronoi-volume-PDF-survey}{\nolinkurl{https://github.com/shitian-0715/Voronoi-volume-PDF-survey}}.

\bibliographystyle{IEEEtran.bst}
\bibliography{MyRef.bib}

\begin{thebibliography}{10}
\providecommand{\url}[1]{#1}
\csname url@samestyle\endcsname
\providecommand{\newblock}{\relax}
\providecommand{\bibinfo}[2]{#2}
\providecommand{\BIBentrySTDinterwordspacing}{\spaceskip=0pt\relax}
\providecommand{\BIBentryALTinterwordstretchfactor}{4}
\providecommand{\BIBentryALTinterwordspacing}{\spaceskip=\fontdimen2\font plus
\BIBentryALTinterwordstretchfactor\fontdimen3\font minus
  \fontdimen4\font\relax}
\providecommand{\BIBforeignlanguage}[2]{{%
\expandafter\ifx\csname l@#1\endcsname\relax
\typeout{** WARNING: IEEEtran.bst: No hyphenation pattern has been}%
\typeout{** loaded for the language `#1'. Using the pattern for}%
\typeout{** the default language instead.}%
\else
\language=\csname l@#1\endcsname
\fi
#2}}
\providecommand{\BIBdecl}{\relax}
\BIBdecl

\bibitem{Okabe2000}
A.~Okabe, B.~Boots, K.~Sugihara, and S.~N. Chiu, \emph{Spatial Tessellations:
  Concepts and Applications of Voronoi Diagrams}, 2nd~ed.\hskip 1em plus 0.5em
  minus 0.4em\relax Wiley, 2000.

\bibitem{Schneider2008}
R.~Schneider and W.~Weil, \emph{Stochastic and Integral Geometry}.\hskip 1em
  plus 0.5em minus 0.4em\relax Berlin, Germany: Springer, 2008.

\bibitem{Chiu2013}
S.~N. Chiu, D.~Stoyan, W.~S. Kendall, and J.~Mecke, \emph{Stochastic Geometry
  and Its Applications}, 3rd~ed.\hskip 1em plus 0.5em minus 0.4em\relax
  Chichester, UK: Wiley, 2013.

\bibitem{pournin2012voronoi}
L.~Pournin and T.~M. Liebling, ``Voronoi diagrams and delaunay triangulations:
  Ubiquitous siamese twins,'' \emph{Documenta Math.}, vol. {Extra Volume ISMP},
  pp. 419--431, 2012.

\bibitem{Hjelle2006}
{\O}.~Hjelle and M.~D{\ae}hlen, \emph{Triangulations and Applications}.\hskip
  1em plus 0.5em minus 0.4em\relax Berlin, Germany: Springer-Verlag, 2006.

\bibitem{Baccelli2009}
F.~Baccelli and B.~B{\l}aszczyszyn, \emph{Stochastic Geometry and Wireless
  Networks, Volume I: Theory}, ser. Foundations and Trends in Networking.\hskip
  1em plus 0.5em minus 0.4em\relax Hanover, MA, USA: Now Publishers, 2009,
  vol.~3, no. 3--4.

\bibitem{Baccelli2010}
------, \emph{Stochastic Geometry and Wireless Networks, Volume II:
  Applications}, ser. Foundations and Trends in Networking.\hskip 1em plus
  0.5em minus 0.4em\relax Hanover, MA, USA: Now Publishers, 2010, vol.~4, no.
  1--2.

\bibitem{Haenggi2012}
M.~Haenggi, \emph{Stochastic Geometry for Wireless Networks}.\hskip 1em plus
  0.5em minus 0.4em\relax Cambridge University Press, 2012.

\bibitem{6042301}
J.~G. Andrews, F.~Baccelli, and R.~K. Ganti, ``A tractable approach to coverage
  and rate in cellular networks,'' \emph{IEEE Trans. Commun.}, vol.~59, no.~11,
  pp. 3122--3134, Nov. 2011.

\bibitem{8358978}
M.~Xia and S.~A{\"i}ssa, ``Unified analytical volume distribution of
  {Poisson--Delaunay} simplex and its application to coordinated multi-point
  transmission,'' \emph{IEEE Trans. Wireless Commun.}, vol.~17, no.~7, pp.
  4912--4921, Jul. 2018.

\bibitem{10720695}
Y.~Li, D.~Guo, L.~Luo, and M.~Xia, ``Air-to-ground communications beyond {5G}:
  {CoMP} handoff management in {UAV} network,'' \emph{IEEE Trans. Wireless
  Commun.}, vol.~23, no.~12, pp. 18\,822--18\,837, Dec. 2024.

\bibitem{11346534}
X.~Fan, W.~Wen, P.~Wu, J.~Zhao, and M.~Xia, ``Air-to-ground communications for
  {I}nternet of {T}hings: {UAV}-based coverage hole detection and recovery,''
  \emph{IEEE Internet Things J.}, vol.~13, no.~6, pp. 12\,228--12\,243, Mar.
  2026.

\bibitem{11311373}
T.~Shi, W.~Wen, P.~Wu, and M.~Xia, ``Vertical heterogeneous networks beyond
  {5G}: {CoMP} coverage enhancement and optimization,'' \emph{IEEE Trans.
  Wireless Commun.}, vol.~25, pp. 9391--9405, 2026.

\bibitem{10531691}
Y.~He, Z.~Li, and Y.~Chen, ``Tractable modeling and performance analysis of
  low-earth orbit satellite constellations,'' \emph{IEEE Internet Things J.},
  vol.~11, no.~17, pp. 28\,297--28\,306, Sep. 2024.

\bibitem{11159552}
C.-S. Choi and F.~Baccelli, ``Stochastic geometry and dynamical system analysis
  of {Walker} satellite constellations,'' \emph{IEEE Trans. Veh. Technol.},
  vol.~75, no.~3, pp. 5127--5132, Mar. 2026.

\bibitem{8533634}
M.~Mozaffari, A.~Taleb Zadeh~Kasgari, W.~Saad, M.~Bennis, and M.~Debbah,
  ``Beyond {5G} with {UAVs}: Foundations of a {3D} wireless cellular network,''
  \emph{IEEE Trans. Wireless Commun.}, vol.~18, no.~1, pp. 357--372, Jan. 2019.

\bibitem{9151343}
Y.~Li, N.~I. Miridakis, T.~A. Tsiftsis, G.~Yang, and M.~Xia, ``Air-to-air
  communications beyond {5G}: A novel {3D} {CoMP} transmission scheme,''
  \emph{IEEE Trans. Wireless Commun.}, vol.~19, no.~11, pp. 7324--7338, Nov.
  2020.

\bibitem{Calka2003}
P.~Calka, ``Precise formulae for the distributions of the principal geometric
  characteristics of the typical cells of a two-dimensional {Poisson--Voronoi}
  tessellation and a {Poisson} line process,'' \emph{Adv. Appl. Probab.},
  vol.~35, no.~3, pp. 551--562, 2003.

\bibitem{xia2026scaleshape}
\BIBentryALTinterwordspacing
T.~Shi and M.~Xia, ``An exact scale--shape factorization of the typical
  {Poisson--Voronoi} cell volume,'' \emph{arXiv preprint arXiv:2607.20680},
  2026. [Online]. Available: \url{https://arxiv.org/abs/2607.20680}
\BIBentrySTDinterwordspacing

\bibitem{Rathie1992}
P.~N. Rathie, ``On the volume distribution of the typical {Poisson--Delaunay}
  cell,'' \emph{J. Appl. Prob.}, vol.~29, no.~3, pp. 740--744, 1992.

\bibitem{Muche1996}
L.~Muche, ``Distributional properties of the three-dimensional
  poisson–delaunay cell,'' \emph{J. Stat. Phys.}, vol.~84, no. 1--2, pp.
  147--167, 1996.

\bibitem{Mishra2023}
A.~Mishra and F.~C. Motta, ``Stability and machine learning applications of
  persistent homology using the {Delaunay--Rips} complex,'' \emph{Front. Appl.
  Math. Stat.}, vol.~9, p. 1179301, 2023.

\bibitem{Zhao2024}
A.~Zhao, X.~Zhao, T.~Gu, Z.~Bi, X.~Sun, C.~Yan, F.~Yang, D.~Zhou, and X.~Zeng,
  ``Exploring high-dimensional search space via {Voronoi} graph traversing,''
  in \emph{Proc. Fortieth Conf. Uncertainty Artificial Intelligence}, ser.
  Proc. Machine Learning Research, vol. 244, 2024, pp. 4219--4236.

\bibitem{Sikorski2024}
\BIBentryALTinterwordspacing
A.~Sikorski and M.~Heida, ``{Voronoi} graph -- improved raycasting and
  integration schemes for high dimensional {Voronoi} diagrams,'' 2024,
  arXiv:2405.10050. [Online]. Available: \url{https://arxiv.org/abs/2405.10050}
\BIBentrySTDinterwordspacing

\bibitem{AHL_2021}
A.~Gusakova and C.~Th{\"a}le, ``The volume of simplices in high-dimensional
  {Poisson--Delaunay} tessellations,'' \emph{Ann. Henri Lebesgue}, vol.~4, pp.
  121--153, 2021.

\bibitem{irlbeck2025pvtd}
\BIBentryALTinterwordspacing
M.~Irlbeck, Z.~Kabluchko, and T.~M{\"u}ller, ``On the shape of the typical
  {P}oisson--{V}oronoi cell in high dimensions,'' \emph{arXiv preprint}, 2025,
  arXiv:2506.02607. [Online]. Available: \url{https://arxiv.org/abs/2506.02607}
\BIBentrySTDinterwordspacing

\bibitem{kiang1966random}
T.~Kiang, ``Random fragmentation in two and three dimensions,'' \emph{Z.
  Astrophys.}, vol.~64, pp. 433--433, 1966.

\bibitem{hinde1980monte}
A.~L. Hinde and R.~E. Miles, ``Monte {C}arlo estimates of the distributions of
  the random polygons of the {V}oronoi tessellation with respect to a {P}oisson
  process,'' \emph{J. Stat. Comput. Simul.}, vol.~10, no. 3--4, pp. 205--223,
  1980.

\bibitem{tanemura2003statistical}
M.~Tanemura, ``Statistical distributions of {P}oisson {V}oronoi cells in two
  and three dimensions,'' \emph{Forma}, vol.~18, no.~4, pp. 221--247, 2003.

\bibitem{weaire1986distribution}
D.~Weaire, J.~P. Kermode, and J.~Wejchert, ``On the distribution of cell areas
  in a {V}oronoi network,'' \emph{Philos. Mag. B}, vol.~53, no.~5, pp.
  L101--L105, 1986.

\bibitem{ferenc2007size}
J.-S. Ferenc and Z.~N{\'e}da, ``On the size distribution of {P}oisson {V}oronoi
  cells,'' \emph{Physica A}, vol. 385, no.~2, pp. 518--526, 2007.

\bibitem{hayen2002areas}
A.~Hayen and M.~P. Quine, ``Areas of components of a {V}oronoi polygon in a
  homogeneous {P}oisson process in the plane,'' \emph{Adv. Appl. Probab.},
  vol.~34, no.~2, pp. 281--291, Jun. 2002.

\bibitem{MucheBallani2011}
L.~Muche and F.~Ballani, ``The second volume moment of the typical cell and
  higher moments of edge lengths of the spatial {P}oisson--{V}oronoi
  tessellation,'' \emph{Monatsh. Math.}, vol. 163, no.~1, pp. 71--80, 2011.

\bibitem{Calka_2002}
P.~Calka, ``The distributions of the smallest disks containing the
  {P}oisson--{V}oronoi typical cell and the {C}rofton cell in the plane,''
  \emph{Adv. Appl. Probab.}, vol.~34, no.~4, pp. 702--717, 2002.

\bibitem{Moller1994}
J.~M{\o}ller, \emph{Lectures on Random Voronoi Tessellations}, ser. Lecture
  Notes in Statistics.\hskip 1em plus 0.5em minus 0.4em\relax New York, NY,
  USA: Springer, 1994, vol.~87.

\bibitem{SCHULTE2012285}
M.~Schulte, ``A central limit theorem for the {P}oisson--{V}oronoi
  approximation,'' \emph{Adv. Appl. Math.}, vol.~49, no.~3, pp. 285--306,
  Sep./Oct. 2012.

\bibitem{Koufos2019}
K.~Koufos and C.~P. Dettmann, ``Distribution of cell area in bounded {P}oisson
  {V}oronoi tessellations with application to secure local connectivity,''
  \emph{J. Stat. Phys.}, vol. 176, no.~5, pp. 1296--1315, Jun. 2019.

\bibitem{MILES197085}
R.~E. Miles, ``On the homogeneous planar {P}oisson point process,'' \emph{Math.
  Biosci.}, vol.~6, pp. 85--127, 1970.

\bibitem{Kumar1992}
S.~Kumar, S.~K. Kurtz, J.~R. Banavar, and M.~G. Sharma, ``Properties of a
  three-dimensional {Poisson--Voronoi} tessellation: A {Monte Carlo} study,''
  \emph{J. Stat. Phys.}, vol.~67, no. 3--4, pp. 523--551, 1992.

\bibitem{Lazar2013}
E.~A. Lazar, J.~K. Mason, R.~D. MacPherson, and D.~J. Srolovitz, ``Statistical
  topology of three-dimensional poisson-voronoi cells and cell boundary
  networks,'' \emph{Phys. Rev. E}, vol.~88, p. 063309, Dec 2013.

\bibitem{Zuyev1992}
S.~A. Zuyev, ``Estimates for distributions of the {Voronoi} polygon's geometric
  characteristics,'' \emph{Random Struct. Algorithms}, vol.~3, no.~2, pp.
  149--162, 1992.

\bibitem{Calka2005}
P.~Calka and T.~Schreiber, ``Limit theorems for the typical {Poisson--Voronoi}
  cell and the {Crofton} cell with a large inradius,'' \emph{Ann. Probab.},
  vol.~33, no.~4, pp. 1625--1642, 2005.

\bibitem{calka2025voronoitail}
\BIBentryALTinterwordspacing
P.~Calka, C.~D’Errico, and N.~Enriquez, ``Sharp asymptotics for the maximal
  distance from the boundary to the nucleus of a typical {Poisson--Voronoi}
  cell,'' \emph{arXiv preprint}, 2025, arXiv:2511.11189. [Online]. Available:
  \url{https://arxiv.org/abs/2511.11189}
\BIBentrySTDinterwordspacing

\bibitem{HugReitznerSchneider2004}
D.~Hug, M.~Reitzner, and R.~Schneider, ``Large {Poisson--Voronoi} cells and
  {Crofton} cells,'' \emph{Adv. Appl. Probab.}, vol.~36, no.~3, pp. 667--690,
  2004.

\bibitem{Mathai93}
A.~M. Mathai, \emph{A Handbook of Generalized Special Functions for Statistics
  and Physical Sciences}.\hskip 1em plus 0.5em minus 0.4em\relax Oxford
  University Press, 1993.

\bibitem{dicenzo1989monte}
S.~B. DiCenzo and G.~K. Wertheim, ``Monte {C}arlo calculation of the size
  distribution of supported clusters,'' \emph{Phys. Rev. B}, vol.~39, no.~10,
  pp. 6792--6792, Apr. 1989.

\bibitem{akhiezer2020classical}
N.~I. Akhiezer, \emph{The Classical Moment Problem and Some Related Questions
  in Analysis}.\hskip 1em plus 0.5em minus 0.4em\relax SIAM, 2020.

\bibitem{8976426}
Y.~Li, M.~Xia, and S.~A{\"i}ssa, ``Coordinated multi-point transmission: A
  {P}oisson--{D}elaunay triangulation based approach,'' \emph{IEEE Trans.
  Wireless Commun.}, vol.~19, no.~5, pp. 2946--2959, May 2020.

\bibitem{Alishahi_Sharifitabar_2008}
K.~Alishahi and M.~Sharifitabar, ``Volume degeneracy of the typical cell and
  the chord length distribution for {Poisson-Voronoi} tessellations in high
  dimensions,'' \emph{Adv. Appl. Probab.}, vol.~40, no.~4, p. 919–938, 2008.

\bibitem{9079449}
C.~Saha and H.~S. Dhillon, ``Load on the typical poisson {V}oronoi cell with
  clustered user distribution,'' \emph{IEEE Wireless Commun. Lett.}, vol.~9,
  no.~9, pp. 1361--1365, Sep. 2020.

\bibitem{8515110}
G.~George, A.~Lozano, and M.~Haenggi, ``Distribution of the number of users per
  base station in cellular networks,'' \emph{IEEE Wireless Commun. Lett.},
  vol.~8, no.~2, pp. 520--523, Apr. 2019.

\bibitem{7389350}
C.-H. Liu and L.-C. Wang, ``Optimal cell load and throughput in green small
  cell networks with generalized cell association,'' \emph{IEEE J. Sel. Areas
  Commun.}, vol.~34, no.~5, pp. 1058--1072, May 2016.

\bibitem{8437232}
H.~Wu, X.~Tao, N.~Zhang, and X.~Shen, ``Cooperative {UAV} cluster-assisted
  terrestrial cellular networks for ubiquitous coverage,'' \emph{IEEE J. Sel.
  Areas Commun.}, vol.~36, no.~9, pp. 2045--2058, Sep. 2018.

\bibitem{11361371}
X.~Chen, J.~Liu, M.~Sheng, and J.~Li, ``Transmissive {RIS}-enabled simultaneous
  coverage for aerial and ground users in cellular networks,'' \emph{IEEE
  Trans. Wireless Commun.}, vol.~25, pp. 10\,576--10\,588, 2026.

\bibitem{9893878}
Y.~Li and M.~Xia, ``Ground-to-air communications beyond {5G}: A coordinated
  multipoint transmission based on {Poisson--Delaunay} triangulation,''
  \emph{IEEE Trans. Wireless Commun.}, vol.~22, no.~3, pp. 1841--1854, Mar.
  2023.

\bibitem{marchand2024cox}
D.~C. Marchand, D.~Coupier, and B.~Henry, ``Line-of-sight {Cox} percolation on
  {Poisson–Delaunay} triangulation,'' \emph{Stochastic Process. Appl.}, vol.
  176, p. 104435, Oct. 2024.

\bibitem{10381632}
X.~Fan, P.~Wu, and M.~Xia, ``Air-to-ground communications beyond {5G}: {UAV}
  swarm formation control and tracking,'' \emph{IEEE Trans. Wireless Commun.},
  vol.~23, no.~7, pp. 8029--8043, Jul. 2024.

\bibitem{6576422}
S.~M. Yu and S.-L. Kim, ``Downlink capacity and base station density in
  cellular networks,'' in \emph{Proc. Int. Symp. Model. Optim. Mobile, Ad Hoc,
  Wireless Netw. (WiOpt)}, 2013, pp. 119--124.

\bibitem{6497002}
S.~Singh and J.~G. Andrews, ``Offloading in heterogeneous networks: Modeling,
  analysis, and design insights,'' \emph{IEEE Trans. Wireless Commun.},
  vol.~12, no.~5, pp. 2484--2497, May 2013.

\bibitem{6684550}
S.~Singh, H.~S. Dhillon, and J.~G. Andrews, ``Joint resource partitioning and
  offloading in heterogeneous cellular networks,'' \emph{IEEE Trans. Wireless
  Commun.}, vol.~13, no.~2, pp. 888--901, Feb. 2014.

\bibitem{7593259}
J.~G. Andrews, T.~Bai, M.~N. Kulkarni, A.~Alkhateeb, A.~K. Gupta, and R.~W.
  Heath, ``Modeling and analyzing millimeter wave cellular systems,''
  \emph{IEEE Trans. Commun.}, vol.~65, no.~1, pp. 403--430, Jan. 2017.

\bibitem{9625443}
M.~S. Zia, D.~M. Blough, and M.~A. Weitnauer, ``On the effects of blockage on
  load modeling in millimeter-wave cellular networks,'' in \emph{Proc. Proc.
  IEEE Veh. Technol. Conf. (VTC'2021-Fall)}, 2021, pp. 1--5.

\bibitem{9676414}
J.-H. Noh, B.~Lee, and S.-J. Oh, ``User-number threshold-based base station
  on/off control for maximizing coverage probability,'' \emph{IEEE Trans. Veh.
  Technol.}, vol.~71, no.~3, pp. 3214--3228, Mar. 2022.

\bibitem{5995290}
P.~C. Pinto, J.~Barros, and M.~Z. Win, ``Secure communication in stochastic
  wireless networks—{Part} {I}: Connectivity,'' \emph{IEEE Trans. Inf.
  Forensics Security}, vol.~7, no.~1, pp. 125--138, Feb. 2012.

\bibitem{6512533}
H.~Wang, X.~Zhou, and M.~C. Reed, ``Physical layer security in cellular
  networks: A stochastic geometry approach,'' \emph{IEEE Trans. Wireless
  Commun.}, vol.~12, no.~6, pp. 2776--2787, Jun. 2013.

\bibitem{9756566}
T.-X. Zheng, X.~Chen, C.~Wang, K.-K. Wong, and J.~Yuan, ``Physical layer
  security in large-scale random multiple access wireless sensor networks: A
  stochastic geometry approach,'' \emph{IEEE Trans. Commun.}, vol.~70, no.~6,
  pp. 4038--4051, Jun. 2022.

\bibitem{9775573}
Z.~Mohamed, A.~Bhowal, and S.~A{\"i}ssa, ``Distance distributions and coverage
  probabilities in {P}oisson--{D}elaunay triangular cells with application to
  coordinated multipoint wireless power transfer,'' \emph{IEEE Trans. Wireless
  Commun.}, vol.~21, no.~11, pp. 9143--9154, Nov. 2022.

\bibitem{Lagarias1999}
J.~C. Lagarias, ``Geometric models for quasicrystals {I}. {Delone} sets of
  finite type,'' \emph{Discrete Comput. Geom.}, vol.~21, no.~2, pp. 161--191,
  1999.

\bibitem{BaakeGrimm2013}
M.~Baake and U.~Grimm, \emph{Aperiodic Order. Volume 1: A Mathematical
  Invitation}.\hskip 1em plus 0.5em minus 0.4em\relax Cambridge, U.K.:
  Cambridge Univ. Press, 2013.

\bibitem{Billingsley1999}
P.~Billingsley, \emph{Convergence of Probability Measures}, 2nd~ed.\hskip 1em
  plus 0.5em minus 0.4em\relax Wiley, 1999.

\end{thebibliography}

\begin{IEEEbiography}
	[{\includegraphics[width=1in, height=1.25in, clip, keepaspectratio]{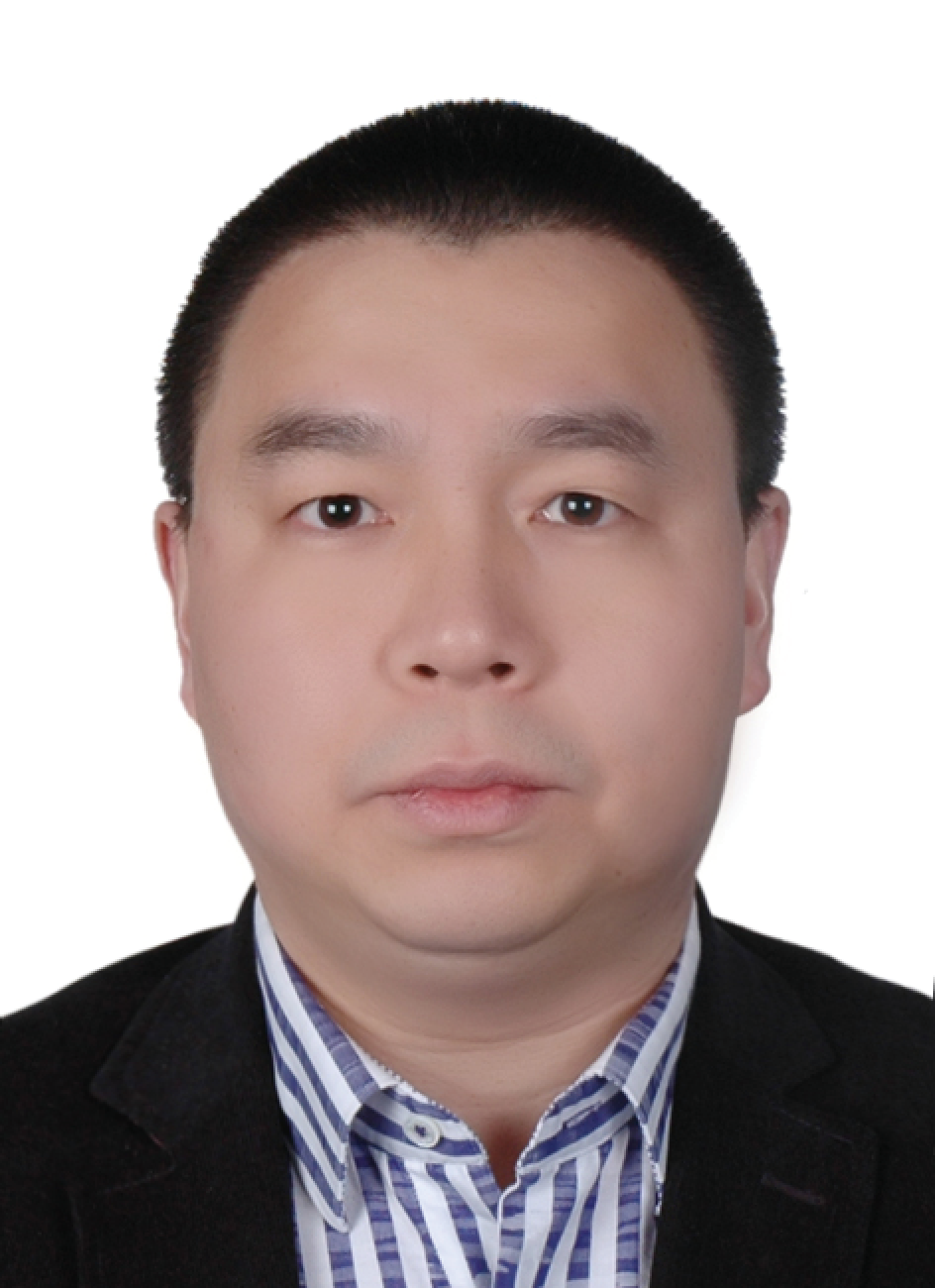}}]{Minghua Xia} (Senior Member, IEEE) received the Ph.D. degree in Telecommunications and Information Systems from Sun Yat-sen University, Guangzhou, China, in 2007.
	
	From 2007 to 2009, he was with the Electronics and Telecommunications Research Institute (ETRI) of South Korea, Beijing R\&D Center, Beijing, China, where he worked as a member and then as a senior member of the engineering staff. From 2010 to 2014, he was in sequence with The University of Hong Kong, Hong Kong, China; King Abdullah University of Science and Technology, Jeddah, Saudi Arabia; and the Institut National de la Recherche Scientifique (INRS), University of Quebec, Montreal, Canada, as a Postdoctoral Fellow. Since 2015, he has been a Professor at Sun Yat-sen University. Since 2019, he has also been an Adjunct Professor with the Southern Marine Science and Engineering Guangdong Laboratory (Zhuhai). His research interests are in the general areas of wireless communications and signal processing.
\end{IEEEbiography}

\begin{IEEEbiography}
	[{\includegraphics[width=1in,height=1.25in, clip, keepaspectratio]{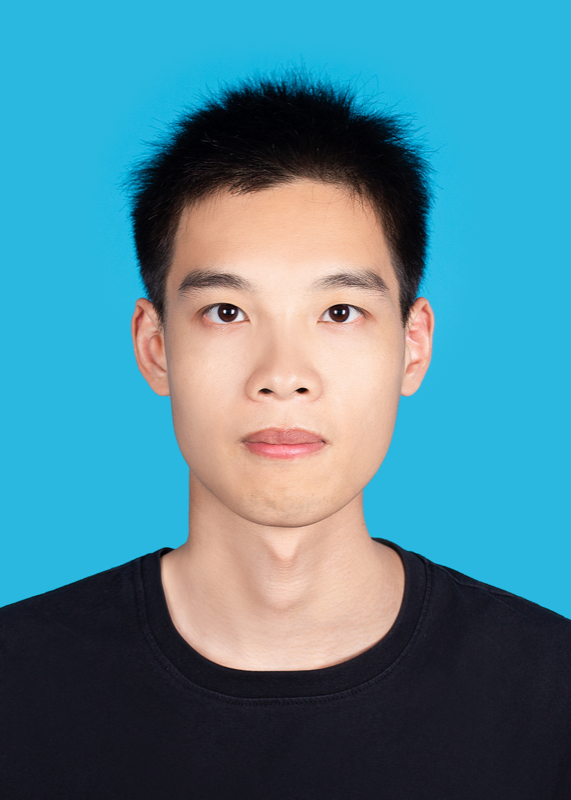}}]{Tian Shi} received the B.Sc. degree in information and computing science in 2019 and the M.Sc. Degree in mathematics in 2023, both from Guilin University of Electronic Technology, Guilin, China. He is currently pursuing a Ph.D. in information and communication engineering at Sun Yat-sen University, Guangzhou, China. His research interests include stochastic geometry and cooperative UAV communications.
\end{IEEEbiography}

\begin{IEEEbiography}
	[{\includegraphics[width=1in, height=1.25in, clip, keepaspectratio]{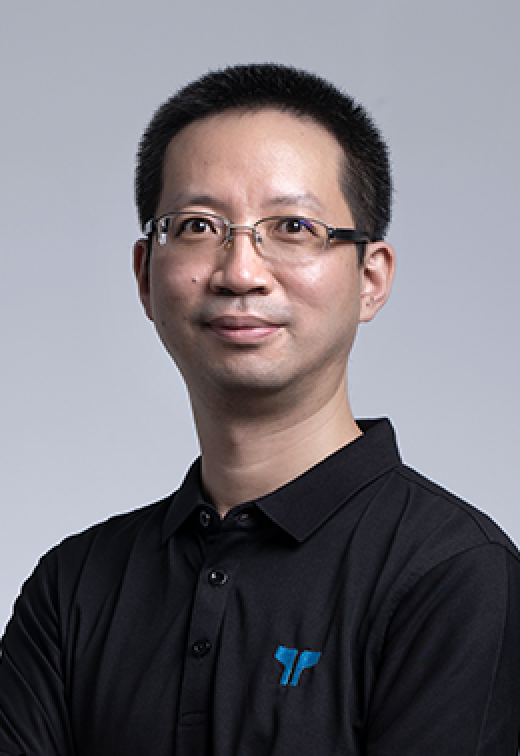}}]{Wenkun Wen} (Member, IEEE) received the Ph.D. degree in Telecommunications and Information Systems from Sun Yat-sen University, Guangzhou, China, in 2007. Since 2020, he has been with Techphant Technologies Co. Ltd., Guangzhou, China, as Chief Engineer.

From 2008 to 2009, he was with the Guangdong-Nortel R\&D center in Guangzhou, China, where he worked as a system engineer on 4G systems. From 2009 to 2012, he worked at the LTE R\&D center of New Postcom Equipment Co. Ltd., Guangzhou, China, where he served as the 4G standard team manager. From 2012 to 2018, he was with the 7th Institute of China Electronic Technology Corporation (CETC) as an expert in wireless communications. From 2018 to 2020, he served as Deputy Director of the 5G Innovation Center at CETC. His research interests include 5G/B5G mobile communications, machine-type communications, narrow-band wireless communications, and signal processing.
\end{IEEEbiography}	
\vfil
\balance
\EOD
\end{document}